\documentclass[12pt]{article}
\pdfoutput=1
\usepackage[authoryear]{natbib}
\usepackage{amsmath,amssymb,amsfonts,amsthm,bbm, stmaryrd, algorithm, float}
\usepackage{epic,eepic,epsfig,longtable}
\usepackage{multirow,verbatim}
\usepackage{array}
\usepackage{verbatim}
\usepackage{epsfig}
\usepackage{setspace}
\usepackage{CJK}
\usepackage{color}

\makeatletter
\def\singlespace{\def\baselinestretch{1}\@normalsize}

\makeatletter
\def\singlespace{\def\baselinestretch{1}\@normalsize}

\numberwithin{equation}{section}

\renewcommand{\hat}{\widehat}

\newcommand{\bfm}[1]{\ensuremath{\mathbf{#1}}}
     \def\bA{\bfm A}          
\def\bb{\bfm b}     \def\bB{\bfm B}          
               
     \def\bD{\bfm D}     \def\cD{{\cal  D}}     
\def\be{\bfm e}     \def\bE{\bfm E}          
\def\bff{\bfm f}              
     \def\bG{\bfm G}          
     \def\bH{\bfm H}          
     \def\bI{\bfm I}

     \def\bM{\bfm M}          
               
          \def\cO{{\cal  O}}     
     \def\bP{\bfm P}          
     \def\bQ{\bfm Q}          
               
     \def\bS{\bfm S}          
               
\def\bu{\bfm u}     \def\bU{\bfm U}          
\def\bv{\bfm v}     \def\bV{\bfm V}          
\def\bw{\bfm w}     \def\bW{\bfm W}          
\def\bx{\bfm x}     \def\bX{\bfm X}          
     \def\bY{\bfm Y}          
     \def\bZ{\bfm Z}          
\def\bzero{\bfm 0}

\newcommand{\bfsym}[1]{\ensuremath{\boldsymbol{#1}}}
       \def \bbeta    {\bfsym{\beta}}

\def \bTheta   {\bfsym{\Theta}}       \def \bLambda  {\bfsym{\Lambda}}
          
\def \bSigma   {\bfsym{\Sigma}}

\renewcommand{\hat}{\widehat}

\def \heps     {\hat{\heps}}

\DeclareMathOperator*{\argmin}{argmin}

\DeclareMathOperator{\Cov}{Cov}

\DeclareMathOperator{\diag}{diag}
\DeclareMathOperator{\E}{E}

\DeclareMathOperator{\rank}{rank}

\DeclareMathOperator{\sgn}{sgn}

\def \Cov   {\mbox{Cov}}

\def \sgn   {\mbox{sgn}}

 at 8truept

\def\today{\ifcase\month\or
  January\or February\or March\or April\or May\or June\or
  July\or August\or September\or October\or November\or December\fi
  \space\number\day, \number\year}

\newdimen\biblioindent\biblioindent=30pt
\newcommand{\beq}  {\begin{equation}}
\newcommand{\eeq}  {\end{equation}}
\newcommand{\beqn} {\begin{eqnarray}}
\newcommand{\eeqn} {\end{eqnarray}}
\newcommand{\beqnn}{\begin{eqnarray*}}
\newcommand{\eeqnn}{\end{eqnarray*}}

\allowdisplaybreaks
\renewcommand{\baselinestretch}{1.33}
\newtheorem{lem}{Lemma}

\newtheorem{thm}{Theorem}
\newcounter{CondCounter}

\newcommand{\ltwonorm}[1]{\lVert#1\rVert_2}

\newcommand{\opnorm}[1]{\lVert#1\rVert_{2}}
\newcommand{\fnorm}[1]{\lVert#1\rVert_F}

\def \col {\text{Col}}

\def \RR	{\mathbb{R}}

\def\R{{\mathbb R}}

\def\E{{\mathbb E}}
\def\S{{\mathbb S}}
\def\P{{\mathbb P}}
\def\diag{{\rm diag}}
\def\Tr{{\rm Tr}}

\def\O{{\cal{O}}}

\def\argmax{{\rm argmax}}
\def\argmin{{\rm argmin}}

\numberwithin{equation}{section}
\theoremstyle{plain}

\newtheorem{defn}{Definition}[section]

\theoremstyle{definition}

\date{}

\begin{document}

	\title{Distributed estimation of principal eigenspaces\thanks{
    The research was supported by NSF grants DMS-1662139 and DMS-1712591 and NIH grant R01-GM072611-12. }}
	\author{Jianqing Fan, Dong Wang, Kaizheng Wang, Ziwei Zhu\\ \normalsize Department of Operations Research and Financial Engineering\\ \normalsize Princeton University}
	
	\maketitle

\begin{abstract}
Principal component analysis (PCA) is fundamental to statistical machine learning.  It extracts latent principal factors that contribute to the most variation of the data. When data are stored across multiple machines, however, communication cost can prohibit the computation of PCA in a central location and distributed algorithms for PCA are thus needed.  This paper proposes and studies a distributed PCA algorithm: each node machine computes the top $K$ eigenvectors and transmits them to the central server; the central server then aggregates the information from all the node machines and conducts a PCA based on the aggregated information.  We investigate the bias and variance for the resulting distributed estimator of the top $K$ eigenvectors.  In particular, we show that for distributions with symmetric innovation, the empirical top eigenspaces are unbiased and hence the distributed PCA is ``unbiased''.  We derive the rate of convergence for distributed PCA estimators, which depends explicitly on the effective rank of covariance, eigen-gap, and the number of machines. We show that when the number of machines is not unreasonably large, the distributed PCA performs as well as the whole sample PCA, even without full access of whole data.  The theoretical results are verified by an extensive simulation study. We also extend our analysis to the heterogeneous case where the population covariance matrices are different across local machines but share similar top eigen-structures.
\end{abstract}


	\section{Introduction}
	
	Principal component analysis (PCA) \citep{Pea01, Hot33} is one of the most fundamental tools in statistical machine learning. The past century has witnessed great efforts on establishing consistency and asymptotic distribution of empirical eigenvalues and eigenvectors. The early classical work of \cite{And63} studied the asymptotic normality of eigenvalues and eigenvectors of sample covariances from multivariate Gaussian distribution with dimension $d$ fixed and sample size $n$ going to infinity. Recent focus moves on to the high-dimensional regimes, i.e., both $n$ and $d$ go to infinity. A partial list of such literatures are \cite{Joh01, BBP05, Pau07, JLu12, JMa09, Ona12, SSZ16, WFa17}. As demonstrated by these papers, asymptotic behaviors of empirical eigenvalues and eigenvectors depend on the scaling of $n, d$ and also the spikiness of the covariance. When $n \ll d$, the empirical leading eigenvector $\hat\bv_1$ is inconsistent in estimating the true top eigenvector $\bv_1$ unless the top eigenvalue $\lambda_1$ diverges fast. This phenomenon inspires another line of research on sparse PCA where certain sparsity on top eigenvectors is imposed to overcome the noise accumulation due to high dimensionality; see e.g., \cite{JLu12, VLe13, SSM13, CMW13}. Besides the asymptotic study, there are also non-asymptotic results on PCA, for example, \cite{Nad08} and \cite{RWa16}.
	
	With rapid developments of information and technology, massive datasets are now ubiquitous. Statistical analysis such as regression  or PCA on such enormous data is unprecedentedly desirable. However, large datasets are usually scattered across distant places such that to fuse or aggregate them is extremely difficult due to communication cost, privacy, data security and ownerships, among others. Consider giant IT companies that collect data simultaneously from places all around the world. Constraints on communication budget and network bandwidth make it nearly impossible to aggregate and maintain global data in a single data center.  Another example is that health records are scattered across many hospitals or countries. It is hard to process the data in a central location due to privacy and ownership concerns.  To resolve these issues, efforts have been made to exploiting distributed computing architectures and developing distributed estimators or testing statistics based on data scattered around different locations. A typical distributed statistical method first calculates local statistics based on each sub-dataset and then combines all the subsample-based statistics to produce an aggregated statistic. Such distributed methods fully adapt to the parallel data collection procedures and thus significantly reduce the communication cost. Many distributed regression methods follow this fashion \citep{ZDW13, CXi14, BFL15, LSL15,BMu17,GLZ17}. The last two papers study distributed kernel regression with spectral regularization using eigen-decomposition of Gram matrices, which is relevant to but different from our distributed PCA.


	Among all the efforts towards creating accurate and efficient distributed statistical methods, there has been rapid advancement on distributed PCA over the past two decades. Unlike the traditional PCA where we have the complete data matrix $\bX\in \RR^{N\times d}$ with $d$ features of $N$ samples at one place, the distributed PCA needs to handle data that are partitioned and stored across multiple servers. There are two data partition regimes: ``horizontal" and ``vertical". In the horizontal partition regime, each server contains all the features of a subset of subjects, while in the vertical partition regime, each server has a subset of features of all the subjects. To conduct distributed PCA in the horizontal regime, \cite{QOS02} proposes that each server computes several top eigenvalues and eigenvectors on its local data and then sends them to the central server that aggregates the information together. Yet there is no theoretical guarantee on the approximation error of the proposed algorithm.  \cite{LBK14}, \cite{KVW14} and \cite{BWZ15} aim to find a good rank-$K$ approximation $\hat{\bX}$ of $\bX$. To assess the approximation quality, they compare $\|\hat{\bX}-\bX\|_F$ against $\min_{rank(\bB)\leq K}\|\bB-\bX\|_F$ and study the excess risk. 
For the distributed PCA in the vertical data partition regime, there is also a great amount of literature, for example, \cite{KHS01}, \cite{LSM11}, \cite{BMo14}, \cite{SAd15}, etc. This line of research is often motivated from sensor networks and signal processing where the vertically partitioned data are common. Our work focuses on the horizontal partition regime, i.e., we have partitions over the samples rather than the features.

Despite these achievements, very few papers establish rigorous statistical error analysis of the proposed distributed PCA methods. To our best knowledge, the only works that provide statistical analysis so far are \cite{EdA10} and \cite{CCH16}. 
To estimate the leading singular vectors of a large target matrix, both papers propose to aggregate singular vectors of multiple random approximations of the original matrix. \cite{EdA10} adopts sparse approximation of the matrix by sampling the entries, while \cite{CCH16} uses Gaussian random sketches. The works are related to ours, since we can perceive sub-datasets in the distributed PCA problem as random approximations.
However, our analysis is more general, since it does not rely on any matrix incoherence assumption as required by \cite{EdA10} and it explicitly characterizes how the probability distribution affects the final statistical error in finite sample error bounds. Besides, our aggregation algorithm is much simpler than the one in \cite{CCH16}. The manuscript \cite{GSS17} came out after we submitted the first draft of our work. The authors focused on estimation of the first principal component rather than the multi-dimensional eigenspaces, based on very different approaches.
	
	We propose a distributed algorithm with only one-shot communication to solve for the top $K$ eigenvectors of the population covariance matrix $\bSigma$ when samples are scattered across $m$ servers. We first calculate for each subset of data $\ell$ its top $K$ eigenvectors $\{\hat{\bV}_K^{(\ell)}= (\hat\bv^{(\ell)}_1, \cdots, \hat\bv^{(\ell)}_K)\}_{\ell=1}^m$ of the sample covariance matrix there, then compute the average of projection matrices of the eigenspaces $\widetilde \bSigma= (1/m)\sum\limits_{i=1}^m \hat{\bV}_K^{(\ell)} \hat{\bV}_K^{(\ell)^T}$, and finally take the top $K$ eigenvectors of $\widetilde\bSigma$ as the final estimator $\widetilde\bV_K=(\tilde{\bv}^{(\ell)}_1, \cdots, \tilde{\bv}^{(\ell)}_K)$. The communication cost of this method is of order $O(mKd)$. We establish rigorous non-asymptotic analysis of the statistical error $\fnorm{\widetilde\bV_K{\widetilde\bV_K}^T- \bV_K\bV_K^T}$, and show that as long as we have a sufficiently large number of samples in each server, $\widetilde\bV_K$ enjoys the same statistical error rate as the standard PCA over the full sample. The eigenvalues of $\bSigma$ are easily estimated once we get good estimators of the eigenvectors, using another round of communication.
%

	The rest of the paper is organized as follows. In Section \ref{sec:2}, we introduce the problem setup of the distributed PCA. In Section \ref{sec:3}, we elucidate our distributed algorithm for estimating the top $K$  eigenvectors. Section \ref{sec:4} develops the statistical error rates of the aggregated estimator.  The results are extended to heterogeneous samples in Section \ref{sec:5}. Finally in Section \ref{sec:6} we present extensive simulation results to validate our theories.
		
	
\section{Problem setup}
\label{sec:2}

We first collect all the notations that will be used. By convention we use regular letters for scalars and bold letters for both matrices and vectors. We denote the set $\{1, 2, 3, ..., d\}$ by $[d]$ for convenience. For two scalar sequences $\{a_n\}_{n \ge 1}$ and $\{b_n\}_{n\ge 1}$, we say $a_n\gtrsim b_n$ ($a_n \lesssim b_n$) if there exists a universal constant $C>0$ such that $a_n\ge Cb_n$ ($a_n\le Cb_n$), and $a_n\asymp b_n$ if both $a_n\gtrsim b_n$ and $a_n\lesssim b_n$ hold. For a random variable $X\in \RR$, we define $\|X\|_{\psi_2}= \sup_{p\ge 1}(\E|X|^p)^{\frac{1}{p}}/\sqrt{p}$ and define $\|X\|_{\psi_1}= \sup_{p\ge 1}(\E|X|^p)^{\frac{1}{p}}/p$. Please refer to \cite{Ver10} for equivalent definitions of $\psi_2$-norm and $\psi_1$-norm. For two random variables $X$ and $Y$, we use $X\overset{d}{=}Y$ to denote that $X$ and $Y$ have identical distributions. Define $\be_i$ to be the unit vector whose components are all zero except that the $i$-th component equals $1$. For $q\geq r$, $\O_{q\times r}$ denotes the space of $q\times r$ matrices with orthonormal columns. For a matrix $\bA\in \RR^{n\times d}$, we use $\fnorm{\bA}$, $\| \bA \|_{*}$ and $\opnorm{\bA}$ to denote the Frobenius norm, nuclear norm and spectral norm of $\bA$, respectively. $\text{Col}(\bA)$ represents the linear space spanned by column vectors of $\bA$. 
We denote the Moore-Penrose pseudo inverse of a matrix $\bA\in \RR^{d \times d}$ by $\bA^{\dag}$. For a symmetric matrix $\bA$, we use $\lambda_j(\bA)$ to refer to its $j$-th largest eigenvalue.

Suppose we have $N$ i.i.d random samples $\{\bX_i  \}_{i=1}^N\subseteq \RR^d$ with $\E \bX_1=\bzero$ and covariance matrix $\E (\bX_1\bX_1^T)= \bSigma$. By spectral decomposition, $\bSigma= \bV\bLambda\bV^T$, where $\bLambda= \diag(\lambda_1, \lambda_2, \cdots, \lambda_d)$ with $\lambda_1 \ge \lambda_2 \ge \cdots \ge \lambda_d$ and $\bV= (\bv_1, \cdots, \bv_d) \in \O_{d\times d}$. For a given $K\in[d]$, let $\bV_K=(\bv_1, \cdots, \bv_K)$. Our goal is to estimate $\col(\bV_K)$, i.e., the linear space spanned by the top $K$ eigenvectors of $\bSigma$. To ensure the identifiability of $\col(\bV_K)$, we assume $\Delta := \lambda_K - \lambda_{K+1}>0$ and define $\kappa := \lambda_1/\Delta$ to be the condition number. Let $r=r(\bSigma):= \Tr(\bSigma)/\lambda_1$ be the effected rank of $\bSigma$.

The standard way of estimating $\col(\bV_K)$ is to use the top $K$ eigenspace of the sample covariance $\hat{\bSigma}=\frac{1}{N}\sum_{i=1}^{N}\bX_i \bX_i^T$. Let $\hat\bSigma= \hat\bV \hat\bLambda {\hat\bV}^T$ be spectral decomposition of $\hat\bSigma$, where $\hat\bLambda= \diag(\hat\lambda_1, \cdots, \hat\lambda_d)$ with $\hat\lambda_1 \ge \cdots \ge \hat\lambda_d$ and $\hat\bV= (\hat\bv_1, \cdots, \hat\bv_d)$. We use the empirical top $K$ engenspace $\col(\hat\bV_K)$, where $\hat\bV_K= (\hat{\bv}_1,\cdots,\hat{\bv}_K)$, to estimate the eigenspace $\col(\bV_K)$. To measure the statistical error, we adopt $\rho(\hat{\bV}_K,\bV_K) :=\|\hat{\bV}_K{\hat{\bV}_K}^T-\bV_K\bV_K^T\|_F$, which is the Frobenius norm of the difference between projection matrices of two spaces and is a well-defined distance between linear subspaces. In fact, $\rho(\bV_K, \hat\bV_K)$ is equivalent to the so-called $\sin\bTheta$ distance. Denote the singular values of $\hat\bV_K^T\bV_K$ by $\{\sigma_i\}_{i=1}^K$ in descending order. Recall that $\bTheta(\hat\bV_K, \bV_K)= \diag(\theta_1, \cdots, \theta_K)$, the principal angles between $\col(\bV_K)$ and $\col(\hat\bV_K)$, are defined as $\diag(\cos^{-1} \sigma_1, \cdots, \cos^{-1}\sigma_K)$. Then we define $\sin\bTheta(\hat\bV_K, \bV_K)$ to be $\diag(\sin\theta_1, \cdots, \sin\theta_K)$. Note that
	\begin{align}
		\rho^2(\bV_K, \hat\bV_K) &
		= \| \bV_K\bV_K^T \|_F^2 + \| \hat\bV_K\hat\bV_K^T \|_F^2 -2\Tr(\bV_K\bV_K^T\hat\bV_K\hat\bV_K^T)
	= 2K- 2\|\hat\bV_K^T\bV_K\|_F^2 \notag	\\&
	= 2\sum\limits_{i=1}^K (1-\sigma_i^2)= 2\sum\limits_{i=1}^K \sin^2\theta_i= 2\|\sin\bTheta(\hat\bV_K, \bV_K)\|_F^2.
		\label{eq:2.1}
	\end{align}
Therefore, $\rho(\bV_K, \hat\bV_K)$ and $\fnorm{\sin\bTheta(\bV_K, \hat\bV_K)}$ are equivalent.


Now consider the estimation of top $K$ eigenspace under the distributed data setting, where our $N=m\cdot n$ samples are scattered across $m$ machines with each machine storing $n$ samples\footnote{Note that here for simplicity we assume the subsample sizes are homogeneous. We can easily extend our analysis to the case of heterogeneous sub-sample sizes with similar theoretical results.}. Application of standard PCA here requires data or covariance aggregation, thus leads to huge communication cost for high-dimensional big data. In addition, for the areas such as genetic, biomedical studies and customer services, it is hard to communicate raw data because of privacy and ownership concerns. To address these problems, we need to avoid naive data aggregation and design a communication-efficient and privacy-preserving distributed algorithm for PCA. In addition, this new algorithm should be statistically accurate in the sense that it enjoys the same statistical error rate as the full sample PCA.

Throughout the paper, we assume that all the random samples $\{\bX_i\}_{i=1}^N$ are i.i.d sub-Gaussian. We adopt the definition of sub-Gaussian random vectors in \cite{KLo17} and \cite{RWa16} as specified below, where $M$ is assumed to be a constant. It is not hard to show that the following definition is equivalent to the definition $\|(\bSigma^{1/2})^{\dag}\bX\|_{\psi_2}\leq M$ used in \cite{Ver10}, \cite{WFa17}, and many other authors.
\begin{defn}
	\label{def:2.1}
	We say the random vector $\bX\in \RR^d$ is sub-Gaussian if there exists $M>0$ such that $\|\bu^T\bX\|_{\psi_2}\leq M \sqrt{\E (\bu^T\bX)^2}$, $\forall \bu\in\R^d$.
\end{defn}

We emphasize here that the global i.i.d assumption on $\{\bX_i\}_{i=1}^N$ can be further relaxed. In fact, our statistical analysis only requires the following three conditions: (i) within each server $\ell$, data are i.i.d.; (ii) across different servers, data are independent; (iii) the covariance matrices of the data in each server $\{\bSigma^{(\ell)}\}_{\ell=1}^m$ share similar top $K$ eigenspaces. We will further study this heterogeneous regime in Section 5. To avoid future confusion, unless specified, we always assume i.i.d. data across servers.

\section{Methodology}
\label{sec:3}

We now introduce our distributed PCA algorithm. For $\ell\in [m]$, let $\{\bX^{(\ell)}_{i}\}_{i=1}^{n}$ denote the samples stored on the $\ell$-th machine. We specify the distributed in Algorithm \ref{algo:1}.

\begin{algorithm}
\caption{Distributed PCA}
\label{algo:1}
\begin{enumerate}
	\item On each server, compute locally the $K$ leading eigenvectors $\hat{\bV}^{(\ell)}_K=(\hat{\bv}^{(\ell)}_{1},\cdots,\hat{\bv}^{(\ell)}_{K})\in\R^{d\times K}$ of the sample covariance matrix $\hat{\bSigma}^{(\ell)}=(1/n)\sum_{i=1}^{n}\bX^{(\ell)}_{i}\bX^{(\ell)^T}_{i}$. Send $\hat{\bV}^{(\ell)}_K$ to the central processor.
	\item On the central processor, compute $\widetilde{\bSigma}=(1/m)\sum_{\ell=1}^{m}\hat{\bV}^{(\ell)}_K\hat{\bV}^{(\ell)^T}_K$, and its $K$ leading eigenvectors $\{\tilde{\bv}_j\}_{j=1}^K$. Output: $\widetilde{\bV}_K=(\widetilde{\bv}_1,\cdots,\widetilde{\bv}_K)\in\R^{d\times K}$.
\end{enumerate}

\end{algorithm}

In other words, each server first calculates the top $K$ eigenvectors of the local sample covariance matrix, and then transmits these eigenvectors $\{\hat\bV_K^{(\ell)}\}_{\ell=1}^m$ to a central server, where the estimators get aggregated. This procedure has similar spirit as distributed estimation based on one-shot averaging in \cite{ZDW13}, \cite{BFL15}, \cite{LSL15}, among others. To see this, we recall the SDP formulation of the eigenvalue problem. Let $\hat{\bV}_K=( \hat{\bv}_1,\cdots, \hat{\bv}_K)$ contain the $K$ leading eigenvectors of $\hat{\bSigma} = \frac{1}{m} \sum_{\ell=1}^{m} \hat\bSigma^{(\ell)}$. Lemma \ref{lem-SDP} in Section 8.2.2 asserts that $\hat\bP_K = \hat{\bV}_K \hat{\bV}_K^T$ solves the SDP:
\begin{equation}\label{SDP-1}
\begin{split}
&\min_{\bP \in S^{d\times d}} -\Tr( \bP^T \hat{\bSigma} )\\
&\text{s.t. } \Tr(\bP)\leq K, \| \bP \|_2 \leq 1, \bP \succeq 0.
\end{split}
\end{equation}
Here $S^{d\times d}$ refers to the set of $d\times d$ symmetric matrices.
In the traditional setting, we have access to all the data, and $\hat\bP_K$ is a natural estimator for $\bV_K \bV_K^T$. In the distributed setting, each machine can only access $\hat\bSigma^{(\ell)}$. Consequently, it solves a local version of \eqref{SDP-1}:
\begin{equation}\label{SDP-2}
\begin{split}
&\min_{\bP \in S^{d\times d}} -\Tr( \bP^T \hat{\bSigma}^{(\ell)} )\\
&\text{s.t. } \Tr(\bP)\leq K, \| \bP \|_2 \leq 1, \bP \succeq 0.
\end{split}
\end{equation}
The optimal solution is $\hat{\bP}_K^{(\ell)} = \hat{\bV}_K^{(\ell)} \hat{\bV}_K^{(\ell) T}$. Since the loss function in \eqref{SDP-1} is the average of local loss functions in \eqref{SDP-2}, we can intuitively average the optimal solutions $\hat{\bP}_K^{(\ell)}$ to approximate $\hat\bP_K$. However, the average $\frac{1}{m} \sum_{\ell=1}^{m} \hat{\bP}_K^{(\ell)}$ may no longer be a rank-$K$ projection matrix. Hence a rounding step is needed, extracting the leading eigenvectors of that average to get a projection matrix.

Here is another way of understanding the aggregation procedure. Given a collection of estimators $\{ \hat{\bV}^{(\ell)}_K \}_{\ell=1}^{m} \subseteq \O_{d\times K}$ and the loss $\rho(\cdot,\cdot)$, we want to find the center $\bU \in \O_{d\times K}$ that minimizes the sum of squared losses $\sum_{\ell=1}^{m} \rho^2( \bU, \hat{\bV}^{(\ell)}_K )$. Lemma \ref{lem-loss} in Section 8.2.2 indicates that $\bU = \widetilde{\bV}_K$ is an optimal solution. Therefore, our distributed PCA estimator $\widetilde{\bV}_K$ is a generalized ``center" of individual estimators.

It is worth noting that in this algorithm, we do not really need to compute $\{\hat{\bSigma}^{(\ell)}\}_{\ell=1}^m$ and $\widetilde{\bSigma}$. $\{\hat\bV_K^{(\ell)}\}_{\ell=1}^m$ and $\widetilde{\bV}_K$ can be derived from top-$K$ SVD of data matrices. This is far more expeditious than the entire SVD and highly scalable, by using, for example, the power method \citep{GVa12}. As regard to the estimation of the top eigenvalues of $\bSigma$, we can send the aggregated eigenvectors $\{\tilde{\bv}_j\}_{j=1}^K$ back to the $m$ servers, where each one computes $\{\lambda_j^{(\ell)}\}_{j=1}^K=\{\tilde{\bv}_j^T\hat{\bSigma}^{(\ell)}\tilde{\bv}_j\}_{j=1}^K$. Then the central server collect all the eigenvalues and deliver the average eigenvalues $\{\tilde{\lambda}_j\}_{j=1}^K=\{ \frac{1}{m}\sum_{\ell=1}^m \lambda_j^{(\ell)} \}_{j=1}^K$ as the estimators of all eigenvalues.

As we can see, the communication cost of the proposed distributed PCA algorithm is of order $O(mKd)$. In contrast, to share all the data or entire covariance, the communication cost will be of order $O(md\min(n, d))$. Since in most cases $K=o(\min(n, d))$, our distributed PCA requires much less communication cost than naive data aggregation.

\section{Statistical error analysis}
\label{sec:4}

Algorithm $\ref{algo:1}$ delivers $\widetilde \bV_K$ to estimate the top $K$ eigenspace of $\bSigma$. In this section we analyze the statistical error  of $\widetilde \bV_K$, i.e., $\rho(\widetilde \bV_K, \bV_K)$. The main message is that $\widetilde\bV_K$ enjoys the same statistical error rate as the full sample counterpart $\hat\bV_K$ as long as the subsample size $n$ is sufficiently large.

We first conduct a bias and variance decomposition of $\rho(\widetilde\bV_K, \bV_K)$, which serves as the key step in establishing our theoretical results. Recall that $\widetilde \bSigma = (1/m) \sum_{\ell=1}^m \hat\bV_K^{(\ell)}\hat\bV_K^{(\ell)T}$ and $\widetilde \bV_K$ consists of the top $K$ eigenvectors of $\widetilde\bSigma$. Define $\bSigma^* := \E (\hat\bV_K^{(\ell)}\hat\bV_K^{(\ell)T})$ and denote its top $K$ eigenvectors by $\bV_K^* = (\bv^*_1, \cdots, \bv^*_K)\in \RR^{d\times K}$. When the number of machines goes to infinity, $\widetilde\bSigma$ converges to $\bSigma^*$, and naturally we expect $\col(\widetilde \bV_K)$ to converge to $\col(\bV^*_K)$ as well. This line of thinking inspires us to decompose the statistical error $\rho(\widetilde \bV_K, \bV_K)$ into the following bias and sample variance terms:
\beq
	\label{eq:4.1}
	\rho(\widetilde\bV_K, \bV_K) \le \underbrace{\rho(\widetilde\bV_K, \bV^*_K)}_{\text{sample variance term}} + \underbrace{\rho(\bV^*_K, \bV_K)}_{\text{bias term}}.
\eeq
The first term is stochastic and the second term is deterministic.
Here we elucidate on why we call $\rho(\widetilde\bV_K, \bV^*_K)$ the sample variance term and $\rho(\bV^*_K, \bV_K)$ the bias term respectively.

\begin{enumerate}
	\item Sample variance term $\rho(\widetilde\bV_K, \bV_K^*)$:
	
	By Davis-Kahan's Theorem (Theorem 2 in \cite{YWS15}) and \eqref{eq:2.1}, we have
\beq
	\label{eq:4.2}
	\rho(\widetilde\bV_K, \bV^*_K) \lesssim \frac{\fnorm{\widetilde\bSigma- \bSigma^*}}{\lambda_K(\bSigma^*)- \lambda_{K+1}(\bSigma^*)}.
\eeq
As we can see, $\rho(\widetilde \bV_K, \bV^*_K)$ depends on how the average $\widetilde\bSigma = \frac{1}{m}\sum\limits_{\ell=1}^m \hat\bV_K^{(\ell)}\hat\bV_K^{(\ell)T}$ concentrates to its mean $\bSigma^*$. This explains why we call $\rho(\widetilde\bV_K, \bV_K^*)$ the sample variance term. We will show in the sequel that for sub-Gaussian random samples, $\{\fnorm{\hat\bV^{(\ell)}_K\hat\bV^{(\ell)T}_K - \bSigma^*} \}_{\ell=1}^m$ and $\fnorm{\widetilde\bSigma- \bSigma^*}$ are sub-exponential random variables and under appropriate regularity assumptions,
\beq
	\label{eq:4.3}
	\left\|\fnorm{\widetilde\bSigma- \bSigma^*} \right\|_{\psi_1} \lesssim \frac{1}{\sqrt{m}} \left\|\fnorm{\hat\bV^{(1)}_K\hat\bV^{(1)T}_K- \bSigma^*} \right\|_{\psi_1}.
\eeq
If we regard $\psi_1$-norm as a proxy for standard deviation, this result is a counterpart to the formula for the standard deviation of the sample mean under the context of matrix concentration.
By \eqref{eq:4.3}, the average of projection matrices $\widetilde\bSigma$ enjoys a similar square-root convergence, so does $\rho(\widetilde\bV_K, \bV^*_K)$.

	\item Bias term $\rho(\bV_K^*, \bV_K)$:
	
The error $\rho(\bV_K^*, \bV_K)$ is deterministic and independent of how many machines we have, and is therefore called  the bias term. We will show this bias term is exactly zero when the random sample has a symmetric innovation (to be defined later).  In general,
we will show that the bias term is negligible in comparison with the sample variance term when the number of nodes $m$ is not unreasonably large.
\end{enumerate}

In the following subsections, we will analyze the sample variance term and bias term respectively and then combine these results to obtain the convergence rate for $\rho(\widetilde\bV_K, \bV_K)$.

\subsection{Analysis of the sample variance term}

%

To analyze $\rho(\widetilde\bV_K, \bV^*_K)$, as shown by \eqref{eq:4.2}, we need to derive the order of the numerator $\fnorm{\widetilde\bSigma- \bSigma^*}$ and denominator $\lambda_K(\bSigma^*)- \lambda_{K+1}(\bSigma^*)$. We first focus on the matrix concentration term $\fnorm{\widetilde\bSigma- \bSigma^*}= \left\| \frac{1}{m}\sum\limits_{\ell=1}^m \Bigl(\hat\bV_K^{(\ell)}\hat\bV_K^{(\ell)T}- \bSigma^* \Bigr) \right\|_F$. Note that $\widetilde \bSigma- \bSigma^*$ is an average of $m$ centered random matrices. To establish the correspondent concentration inequality, we first investigate each individual term in the average, i.e., $\hat\bV_K^{(\ell)}\hat\bV_K^{(\ell)T}- \bSigma^*$ for $\ell\in [m]$. In the following lemma, we show that when random samples are sub-Gaussian, $\fnorm{\hat\bV_K^{(\ell)}\hat\bV_K^{(\ell)T}- \bSigma^*}$ is sub-exponential and we can give an explicit upper bound of its $\psi_1-$norm.

\begin{lem}
	\label{lem:1}
	Suppose that on the $\ell$-th server we have  $n$ i.i.d. sub-Gaussian random samples $\{\bX_i\}_{i=1}^n$ in $\RR^d$ with covariance matrix $\bSigma$.
There exists a constant $C>0$ such that when $n\geq r$,
	\[
		\left\| \fnorm{\hat\bV^{(\ell)}_K\hat\bV^{(\ell)T}_K- \bSigma^*} \right \|_{\psi_1} \le C \kappa \sqrt{\frac{K r}{n}}.
	\]
\end{lem}

Note that here we use the Frobenius norm to measure the distance between two matrices. Therefore, it is equivalent to treat $\{\hat\bV^{(\ell)}_K\hat\bV^{(\ell)T}_K\}_{\ell=1}^K$ and $\bSigma^*$ as $d^2-$dimensional vectors and apply the concentration inequality for random vectors to bound $\fnorm{\widetilde\bSigma- \bSigma^*}$. As we will demonstrate in the proof of Theorem \ref{thm:1}, $\left \| \fnorm{\widetilde\bSigma- \bSigma^*} \right\|_{\psi_1} \lesssim \frac{1}{\sqrt{m}} \left\| \fnorm{\hat\bV^{(\ell)}_K\hat\bV^{(\ell)T}_K- \bSigma^*} \right\|_{\psi_1} $.

With regard to $\lambda_K(\bSigma^*)- \lambda_{K+1}(\bSigma^*)$, when the individual node has enough samples, $\hat\bV^{(\ell)}_K$ and $\bV_K$ will be close to each other and so will $\bSigma^*= \E( \hat\bV^{(\ell)}_K \hat \bV^{(\ell)T}_K)$ and $\bV_K\bV_K^T$. Given $\lambda_K(\bV_K\bV_K^T)=1$ and $\lambda_{K+1}(\bV_K\bV_K^T)=0$, we accordingly expect $\lambda_K(\bSigma^*)$ and $\lambda_{K+1}(\bSigma^*)$ be separated by a positive constant as well.

All the arguments above lead to the following theorem on $\rho(\widetilde\bV_K, \bV^*_K)$.

\begin{thm}
	\label{thm:1}
	Suppose $\bX_1, \cdots, \bX_N$ are i.i.d. sub-Gaussian random vectors in $\RR^d$ with covariance matrix $\bSigma$ and they are scattered across $m$ machines. If $n\geq r$ and $\| \bSigma^* -\bV_K \bV_K^T\|_2 \le 1/4$, then
	\[
		\left\| \rho(\widetilde\bV_K, \bV_K^*) \right\|_{\psi_1} \le C\kappa \sqrt{\frac{Kr}{N}},
	\]
	where $C$ is some universal constant.
\end{thm}



\subsection{Analysis of the bias term}
In this section, we study the bias term $\rho(\bV^*_K, \bV_K)$ in \eqref{eq:4.1}. We first focus on a special case where the bias term is exactly zero. For a random vector $\bX$ with covariance $\bSigma=\bV \bLambda \bV^T$, let $\bZ= \bLambda^{-\frac{1}{2}}\bV^T\bX$. We say $\bX$ has symmetric innovation if $\bZ \overset{d}{=}(\bI_d-2\be_j\be_j^T)\bZ,~~\forall j\in[d]$. In other words, flipping the sign of one component of $\bZ$ will not change the distribution of $\bZ$. Note that if $\bZ$ has density, this is equivalent to say that its density function has the form $p( |z_1|, |z_2|,\cdots, |z_d| )$. All elliptical distributions centered at the origin belong to this family. In addition, if $\bZ$ has symmetric and independent entries, $\bX$ has also symmetric innovation. It turns out that when the random samples have symmetric innovation, $\bSigma^*:= \E (\hat\bV^{(\ell)}_K\hat\bV^{(\ell)T}_K)$ and $\bSigma$ share exactly the same set of eigenvectors. When we were finishing the paper, we noticed that \cite{CCH16} had independently established a similar result for the Gaussian case.

\begin{defn}\label{defn-unbiasedness}
Let $\mathcal{V}$ be a $K$-dimensional linear subspace of $\R^d$. For a subspace estimator represented by $\hat\bV \in \O_{d\times K}$, we say it is {\bf unbiased} for $\mathcal{V}$ if and only if the top $K$ eigenspace of $\E (\hat\bV\hat\bV^T)$ is $\mathcal{V}$.
\end{defn}

If $\hat{\bV}_K^{(\ell)}$ is unbiased for $\col(\bV_K)$, then $\rho(\bV^*_K, \bV_K) \allowbreak =0$ and we will only have the sample variance term in \eqref{eq:4.1}. In that case, aggregating $\{ \hat{\bV}_K^{(\ell)} \}_{\ell=1}^m$ reduces variance and yields a better estimator $\widetilde{\bV}_K$. Theorem \ref{thm:2} shows that this is the case so long as the distribution has symmetric innovation and the sample size is large enough.

\begin{thm}
\label{thm:2}
Suppose on the $\ell$-th server we have $n$ i.i.d. random samples $\{\bX_i\}_{i=1}^n$ with covariance $\bSigma$. If $\{\bX_i\}_{i=1}^n$ have symmetric innovation, then $\bV^T\bSigma^*\bV$ is diagonal, i.e., $\bSigma^*$ and $\bSigma$ share the same set of eigenvectors.  Furthermore, if $
\| \bSigma^* -  \bV_K \bV_K^T \|_2 < 1/2$, then $\{ \hat{\bV}_K^{(\ell)} \}_{\ell=1}^m$ are unbiased for $\col( \bV_K )$ and $\rho(\bV^*_K, \bV_K)=0$.
\end{thm}
It is worth pointing out that distributed PCA is closely related to aggregation of random sketches of a matrix \citep{HMT11,TYU16}. To approximate the subspace spanned by the $K$ leading left singular vectors of a large matrix $\bA \in \R^{d_1 \times d_2}$, we could construct a suitable random matrix $\bY \in \R^{d_2\times n}$ with $n\geq K$, and use the left singular subspace of $\bA \bY \in \R^{d_1 \times n}$ as an estimator. $\bA \bY$ is called a random sketch of $\bA$. It has been shown that to obtain reasonable statistical accuracy, $n$ can be much smaller than $\min(d_1, d_2)$ as long as $\bA$ is approximately low rank. Hence it is much cheaper to compute SVD on $\bA \bY$ than on $\bA$.
When we want to aggregate a number of such subspace estimators, a smart choice of the random matrix ensemble for $\bY$ is always preferable.
It follows from Theorem \ref{thm:2} that if we let $\bY$ have i.i.d. columns from a distribution with symmetric innovation (e.g., Gaussian distribution or independent entries), then the subspace estimators are unbiased, which facilitates aggregation.

Here we explain why we need the condition $\| \bSigma^* - \bV_K \bV_K^T \|_2 <1/2$ to achieve zero bias.  First of all, the condition is similar to a bound on the ``variance'' of the random matrix $\hat\bV_K^{(\ell)}$ whose covariance $\bSigma^*$ is under investigation.   As demonstrated above, with the symmetric innovation, $\bSigma^*$ has the same set of eigenvectors as $\bSigma$, but we still cannot guarantee that the top $K$ eigenvectors of $\bSigma^*$ match with those of $\bSigma$. For example, the $(K+1)$-th eigenvector of $\bSigma$ might be the $K$-th eigenvector of $\bSigma^*$. In order to ensure the top $K$ eigenspace of $\bSigma^*$ is exactly the same as that of $\bSigma$, we require $\hat\bV^{(\ell)}_K$ to not deviate too far from $\bV_K$ so that $\bSigma^*$ is close enough to  $\bV_K\bV_K^T$. Both Theorems \ref{thm:1} and \ref{thm:2} require control of $\| \bSigma^* -  \bV_K \bV_K^T \|_2$, which will be studied shortly.

For general distributions, the bias term is not necessarily zero. However, it turns out that when the subsample size is large enough, the bias term $\rho(\bV^*_K, \bV_K)$ is of high-order compared with the statistical error of $\hat\bV^{(\ell)}_K$ on the individual subsample. By the decomposition \eqref{eq:4.1} and Theorem \ref{thm:1}, we can therefore expect the aggregated estimator $\widetilde\bV_K$ to enjoy sharper statistical error rate than PCA on the individual subsample. In other words, the aggregation does improve the statistical efficiency. A similar phenomenon also appears in statistical error analysis of the average of the debiased Lasso estimators in \cite{BFL15} and \cite{LSL15}. Recall that in sparse linear regression, the Lasso estimator $\hat\bbeta$ satisfies that $\ltwonorm{\hat\bbeta- \bbeta^*}= O_P(\sqrt{s\log d/n})$, where $\bbeta^*$ is the true regression vector, $s$ is the number of nonzero coefficients of $\bbeta^*$ and $d$ is the dimension. The debiasing step reduces the bias of $\hat\bbeta$ to the order $O_P(s\log d/n)$, which is negligible when $m$ is not too large, compared with the statistical error of $\hat\bbeta$ and thus enables the average of the debiased Lasso estimators to enhance the statistical efficiency.

Below we present Lemma \ref{lem:2}, a high-order Davis-Kahan theorem that explicitly characterizes the linear term and high-order error on top $K$ eigenspace due to matrix perturbation. 
This is a genuine generalization of the former high-order perturbation theorems 
on a single eigenvector, e.g., Lemma 1 in \cite{KUt01} and Theorem 2 in \cite{EdA10}.
An elegant result on eigenspace perturbation is Lemma 2 in \cite{KLo16}. Our error bound uses Frobenius norm while theirs uses spectral norm. Besides, when the top $K$ eigenspace is of interest, the upper bound in Lemma 2 in \cite{KLo16} contains an extra factor $1+ ( \lambda_1 - \lambda_K )/\Delta$. Hence we have better dependence on problem parameters.
Other related works in the literature consider asymptotic expansions of perturbation \citep{Kat66,Vac94,Xu02}, and singular space of a matrix contaminated by Gaussian noise \citep{Wan15}.
Our result is both non-asymptotic and deterministic. It serves as the core of bias analysis.

\begin{lem}
	\label{lem:2}
Let $\bA, \hat{\bA}\in\R^{d\times d}$ be symmetric matrices with eigenvalues $\lambda_1\geq\cdots\geq\lambda_d$, and $\hat{\lambda}_1\geq\cdots\geq\hat{\lambda}_d$, respectively. Let $\{\bu_j\}_{j=1}^d$, $\{\hat{\bu}_j\}_{j=1}^d$ be two orthonormal bases of $\R^d$ such that $\bA\bu_j=\lambda_j\bu_j$ and $\hat{\bA}\hat{\bu}_j=\hat{\lambda}_j\hat{\bu}_j$ for all $j\in[d]$. Fix $s\in\{0,1,\cdots,d-K\}$ and assume that $\Delta=\min\{ \lambda_s-\lambda_{s+1},\lambda_{s+K}-\lambda_{s+K+1} \}>0$, where $\lambda_0=+\infty$ and $\lambda_{d+1}=-\infty$. Define $\bU=(\bu_{s+1},\cdots,\bu_{s+K})$, $\hat{\bU}=(\hat{\bu}_{s+1},\cdots,\hat{\bu}_{s+K})$.
Define $\bE=\hat{\bA}-\bA$, $S=\{s+1,\cdots,s+K\}$, $\bG_j=\sum_{i\notin S}(\lambda_i-\lambda_{s+j})^{-1}\bu_i\bu_i^T$ for $j\in [K]$, and
\[
f:\mathbb{R}^{d\times K}\rightarrow\mathbb{R}^{d\times K},~~ (\bw_1,\cdots,\bw_K)\mapsto(-\bG_{1}\bw_1,\cdots,-\bG_{K}\bw_K).
\]
When $\varepsilon=\|\bE\|_{2}/\Delta \leq 1/10$, we have
	\begin{align*}
	&\|  \hat{\bU} \hat{\bU}^T  - \bU \bU^T - [ f(\bE\bU)\bU^T+\bU f(\bE \bU)^T]\|_F
	\leq 24 \sqrt{K} \varepsilon^2.
	\end{align*}
\end{lem}


Similar to Taylor expansion, the difference is decomposed into the linear leading term and residual of higher order with respect to the perturbation. Here we only present a version that is directly applicable to bias analysis. Stronger results are summarized in Lemma \ref{lem:DK-strong} in Section 8.2.2, which may be of independent interest in perturbation analysis of spectral projectors.

Now we apply Lemma~\ref{lem:2} to the context of principal eigenspace estimation. Let $\bA= \bSigma$, $\hat \bA= \hat\bSigma ^{(1)} $ and $S=[K]$. It thus follows that $\bU=\bV_K$, $\hat\bU= \hat\bV_K^{(1)}$ and $\bE= \hat\bSigma ^{(1)} - \bSigma$.
From the second inequality in Lemma \ref{lem:2} we can conclude that the bias term $\rho(\bV_K^*, \bV_K)$ is a high-order term compared with the linear leading term. More specifically, the Davis-Kahan theorem helps us control the bias as follows: 
\begin{equation*}
\rho(\bV^*_K, \bV_K) \lesssim \fnorm{\bSigma^*- \bV_K\bV_K^T}  = \fnorm{\E[\hat\bV_K^{(1)}\hat\bV_K^{(1) T}- \bV_K\bV_K^T]}.
\end{equation*}
By the facts that $\E ( \bE ) = 0$ and $f$ is linear, we have
$$
\rho(\bV^*_K, \bV_K)	= \fnorm{\E [\hat\bV_K^{(1)}\hat\bV_K^{(1) T} - (\bV_K\bV_K^T+ f(\bE\bV_K)\bV_K^T + \bV_Kf(\bE\bV_K)^T)]}.
$$
By Jensen's inequality, the right hand side above is further bounded by
\begin{align}
\E \fnorm{\hat\bV_K^{(1)}\hat\bV_K^{(1)T}- (\bV_K\bV_K^T+ f(\bE\bV_K)\bV_K^T + \bV_Kf(\bE\bV_K)^T)}.
\label{eqn-bias}
\end{align}
When $n$ is large enough, the typical size of $\varepsilon = \| \bE \|_2/\Delta$ is small, and Lemma \ref{lem:1} controls it tail and all of the moments. Together with Lemma \ref{lem:2}, this fact implies that \eqref{eqn-bias} has roughly the same order as $\sqrt{K} \cdot \E \varepsilon^2$, which should be much smaller than the typical size of $\sqrt{K} \varepsilon$, i.e. the upper bound for $\rho(\hat\bV^{(1)}_K, \bV_K)$ given by Davis-Kahan theorem.
The following theorem makes our hand-waving analysis rigorous.

\begin{thm}
	\label{thm:3}
There are constants $C_1$ and $C_2$ such that when $n\geq r$,
	\[
	\rho(\bV^*_K, \bV_K) \leq C_1 \| \bSigma^* - \bV_K \bV_K^T \|_F
	\leq C_2 \kappa^2 \sqrt{K} r / n.
	\]
\end{thm}
As a by-product, we get
$\| \bSigma^* - \bV_K \bV_K^T \|_2 \lesssim \kappa^2 \sqrt{K} r / n$. Hence when $n\geq C \kappa^2 \sqrt{K}r$ for some large enough $C$, the assumptions in Theorems \ref{thm:1} and \ref{thm:2} on $\| \bSigma^* - \bV_K \bV_K^T \|_2$ are guaranteed to hold.

\subsection{Properties of distributed PCA}

We now combine the results we obtained in the previous two subsections to derive the statistical error rate of $\widetilde \bV_K$. We first present a theorem under the setting of global i.i.d. data and discuss its optimality.

\begin{thm}
	\label{thm:4}
	Suppose we have $N$ i.i.d. sub-Gaussian random samples with covariance $\bSigma$. They are scattered across $m$ servers, each of which stores $n$ samples. There exist constants $C,C_1,C_2,C_3$ and $C_4$ such that the followings hold when $n \ge C \kappa^2 \sqrt{K} r$.
	\begin{enumerate}
		\item Symmetric innovation:
	\beq
		\label{eq:4.7}
		\left\| \rho(\widetilde\bV_K, \bV_K) \right\|_{\psi_1} \le C_1 \kappa \sqrt{\frac{Kr}{N}}.
	\eeq
		\item General distribution:
	\beq
		\label{eq:4.8}
		\left\| \rho(\widetilde\bV_K, \bV_K) \right\|_{\psi_1} \le C_1 \kappa \sqrt{\frac{Kr}{N}}+
		C_2 \kappa^2 \frac{  \sqrt{K} r}{n}.
	\eeq
Furthermore, if we further assume $m \le C_3 n/(\kappa^2 r)$,
\beq
\label{eq:4.9}
\left\| \rho(\widetilde\bV_K, \bV_K) \right\|_{\psi_1} \le C_4 \kappa \sqrt{\frac{Kr}{N}}.
\eeq
\end{enumerate}
\end{thm}
As we can see, with appropriate scaling conditions on $n$, $m$ and $d$, $\widetilde\bV_K$ can achieve the statistical error rate of order $\kappa \sqrt{Kr/N}$.
The result is applicable to the whole sample or traditional PCA, in which $m=1$.
Hence the distributed PCA and the traditional PCA share the same error bound as long as the technical conditions are satisfied. 

In the second part of Theorem \ref{thm:4}, the purpose of setting restrictions on $n$ and $m$ is to ensure that the distributed PCA algorithm delivers the same statistical rate as the centralized PCA which uses all the data. In the boundary case where $n \asymp \kappa^2 \sqrt{K} r$, the bias of the local empirical eigenspace is of constant order. Since our aggregation cannot kill bias, there is no hope to achieve the centralized rate unless the number of machines is of constant order so that the centralized PCA has constant error too. Besides, our result says that when $n$ is large, we can tolerate more data splits (larger $m$) for achieving the centralized statistical rate.

We now illustrate our result through a simple spiked covariance model introduced by \cite{Joh01}. Assume that $\bLambda=\diag(\lambda,  \underbrace{1, \cdots, 1}_{d-1})$,
where $\lambda>1$, and we are interested in the first eigenvector of $\bSigma$.
Note that $K=1$, $r= \text{Tr}(\bSigma)/\opnorm{\bSigma} = (\lambda+d-1)/\lambda \asymp d/\lambda$  when $\lambda=O(d)$, and $\kappa = \lambda/(\lambda-1) \asymp 1$.
It is easy to see from \eqref{eq:4.7} or \eqref{eq:4.9} that
\[
\left\| \rho(\widetilde\bV_1, \bV_1) \right\|_{\psi_1}
\lesssim \kappa \sqrt{\frac{r}{N}} \lesssim
\sqrt{\frac{d}{N\lambda}}.
\]
Without loss of generality, we could always assume that the direction of $\widetilde{\bV}_1$ is chosen such that $\widetilde{\bV}_1^T \bV_1 \geq 0$, i.e. $\widetilde{\bV}_1$ is aligned with $\bV_1$. Note that
\[
\begin{aligned}
\rho^2(\widetilde\bV_1, \bV_1) & = \fnorm{\widetilde\bV_1\widetilde\bV_1^T- \bV_1\bV_1^T}^2 = 2(1-\widetilde\bV_1^T\bV_1)(1+\widetilde\bV_1^T\bV_1)
\ge 2(1-\widetilde\bV_1^T\bV_1)= \ltwonorm{\widetilde\bV_1-\bV_1}^2.
\end{aligned}
\]
Hence
\beq
\label{eq:4.11}
\E \ltwonorm{\widetilde\bV_1-\bV_1}^2 \lesssim
\left\| \rho(\widetilde\bV_1, \bV_1) \right\|_{\psi_1}^2 \lesssim
\frac{d}{N\lambda}.
\eeq
We now  compare this rate with the previous results under the spiked model. In \cite{PJo12}, the authors derived the $\ell_2$ risk of the empirical eigenvectors when random samples are Gaussian. It is not hard to derive from Theorem 1 therein that given $N$ i.i.d $d$-dimensional Gaussian samples, when $N, d$ and $\lambda$ go to infinity,
\[
	\E \ltwonorm{\hat\bV_1- \bV_1}^2 \asymp \frac{d}{N\lambda},
\]
where $\hat\bV_1$ is the empirical leading eigenvector with $\hat\bV_1^T \bV_1 \ge 0$. We see from \eqref{eq:4.11} that the aggregated estimator $\widetilde\bV_1$ performs as well as the full sample estimator $\hat\bV_1$ in terms of the mean squared error.  See \cite{WFa17} for generalization of the results for spiked covariance.

In addition, our result is consistent with the minimax lower bound developed in \cite{CMW13}. For $\lambda>0$ and fixed $c \geq 1$, define
\[
	\Theta=\{\bSigma \text{ is symmetric and }\bSigma \succeq 0: \lambda+1 \le \lambda_K \le \lambda_1 \le c \lambda+1, \lambda_j=1 \text{ for } K+1 \leq  j \leq d\}.
\]
Assume that $K\leq d/2$ and $1\lesssim d/\lambda \lesssim N$.
Theorem 8 in \cite{CMW13} shows that under the Gaussian distribution with $\bSigma\in \Theta$, the minimax lower bound of $\E \rho^2 (\hat\bV, \bV_K)$ satisfies
\begin{align}
\inf\limits_{\widehat\bV} \sup\limits_{\bSigma\in \Theta} \E \rho^2(\hat\bV, \bV_K) \gtrsim \min
\left\{
K, (d-K), \frac{K(\lambda+1)(d-K)}{N\lambda^2} \right\}
\gtrsim \frac{Kd}{N\lambda} .
\label{pca-lower-bound}
\end{align}
Based on $r= \text{Tr}(\bSigma)/\opnorm{\bSigma} \leq (cK\lambda+d)/(c\lambda+1) \lesssim Kd/\lambda$ and $\kappa \leq c \lesssim 1$, our \eqref{eq:4.7} gives an upper bound
\[
\E \rho^2(\widetilde\bV_1, \bV_1) \lesssim \kappa^2 \frac{Kr}{n} \lesssim \frac{Kd}{N\lambda},
\]
which matches the lower bound in \eqref{pca-lower-bound}.

Although the upper bound $\kappa \sqrt{Kr/N}$ established in Theorem \ref{thm:4} is optimal in the minimax sense as discussed above, the non-minimax risk of empirical eigenvectors can be improved when the condition number $\kappa$ is large.  See \cite{VLe13}, \cite{KLo16} and \cite{RWa16} for sharper results. We use \eqref{eq:4.7} as a benchmark rate for the centralized PCA only for the sake of simplicity.

Notice that in Theorem \ref{thm:4}, the prerequisite for $\widetilde\bV_K$ to enjoy the sharp statistical error rate is a lower bound on the subsample size $n$, i.e.,
\beq
	\label{eq:4.12}
	n \gtrsim \kappa^2 \sqrt{K} r.
\eeq
As in the remarks after Lemma \ref{lem:2}, this is the condition we used to ensure closeness between $\bSigma^*$ and $\bV_K \bV_K^T$.
It is natural to ask whether this required sample complexity is sharp, or in other words, is it possible for $\widetilde\bV_K$ to achieve the same statistical error rate with a smaller sample size on each machine? The answer is no. The following theorem presents a distribution family under which $\col(\widetilde\bV_K)$ is even perpendicular to $\col(\bV_K)$ with high probability when $n$ is smaller than the threshold given in \eqref{eq:4.12}. This means that having a smaller sample size on each machine is too uninformative such that the aggregation step completely fails in improving estimation consistency.

\begin{thm}
	\label{thm:5}
Consider a Bernoulli random variable $W$ with $P(W=0)=P(W=1)=1/2$, a Rademacher random variable $P(Y=1)=P(Y=-1)=1/2$, and a random vector $\bZ\in \R^{d-1}$ that is uniformly distributed over the $(d-1)$-dimensional unit sphere. For $\lambda\geq 2$, we say a random vector $\bX\in \R^d$ follows the distribution $\cD(\lambda)$ if
	\[
		\bX\overset{d}{=}\left(
			\begin{array}{c}
				{\bf 1}_{\{W=0\}}\sqrt{2\lambda}Y \\
				{\bf 1}_{\{W=1\}}\sqrt{2(d-1)}\bZ
			\end{array}
			\right).
	\]
	Now suppose we have $\{\bX_i\}_{i=1}^N$ as $N$ i.i.d. random samples of $\bX$. They are stored across $m$ servers, each of which has $n$ samples. When $32 \log d\le n\le (d-1)/(3\lambda)$, we have
	\[
		P(\widetilde{\bV}_1\perp \bV_1)\geq
			\begin{cases}
				&1-d^{-1},~~\text{if}~~m\leq d^3,\\
				&1-e^{-d/2},~~\text{if}~~m> d^3.
			\end{cases}	
	\]
\end{thm}


It is easy to verify that $\cD(\lambda)$ is symmetric, sub-Gaussian and satisfies $\E \bX=\bzero$ and $\E ( \bX\bX^T )= \diag(\lambda, 1, \cdots, 1)$. Besides, $\kappa = \lambda/(\lambda-1) \asymp 1$ and $r = (\lambda+d-1)/\lambda = d/\lambda + 1 - \lambda^{-1} \asymp d/\lambda$ when $2\leq \lambda \lesssim d$. According to \eqref{eq:4.12}, we require $n \gtrsim d/\lambda$ to achieve the rate as demonstrated in \eqref{eq:4.7}. Theorem \ref{thm:5} shows that if we have fewer samples than this threshold, the aggregated estimator $\widetilde\bV_1$ will be perpendicular to the true top eigenvector $\bV_1$ with high probability. Therefore, our lower bound for the subsample size $n$ is sharp.

\section{Extension to heterogeneous samples}
\label{sec:5}

We now  relax global $i.i.d.$ assumptions in the previous section to the setting of heterogeneous covariance structures across servers. Suppose data on the server $\ell$ has covariance matrix $\bSigma^{(\ell)}$, whose top $K$ eigenvalues and eigenvectors are denoted by $\{\lambda^{(\ell)}_k\}_{k=1}^K$ and $\bV^{(\ell)}_K=(\bv^{(\ell)}_1, \cdots, \bv^{(\ell)}_K)$ respectively.
We will study two specific cases of heterogeneous covariances: one requires all covariances to share exactly the same principal eigenspaces, while the other considers the heterogeneous factor models with common factor eigen-structures.

\subsection{Common principal eigenspaces}
We assume that $\{\bSigma^{(\ell)}\}_{\ell=1}^m$ share the same top $K$ eigenspace, i.e. there exists some $\bV_K\in \O_{d\times K}$ such that $\bV_K^{(\ell)}\bV_K^{(\ell)T}= \bV_K\bV_K^T$ for all $\ell\in[m]$.
The following theorem can be viewed as a generalization of Theorem \ref{thm:4}.
\begin{thm}
	\label{thm:6}
	Suppose we have in total $N$ sub-Gaussian samples scattered across $m$ servers, each of which stores $n$ i.i.d. samples with covariance $\bSigma^{(\ell)}$. Assume that $\{\bSigma^{(\ell)}\}_{\ell=1}^m $ share the same top $K$ eigenspace. For each $\ell\in [m]$, let $S_{\ell}= \kappa_{\ell} \sqrt{\frac{K r_{\ell}}{N}}$ and $B_{\ell}=
		  \frac{\kappa^2_{\ell}\sqrt{K} r_{\ell} }{n}$,
where $r_{\ell} :=\Tr( \bSigma^{(\ell)} )/\lambda_1^{(\ell)}$ and $\kappa_{\ell}:=\lambda^{(\ell)}_1/( \lambda^{(\ell)}_K - \lambda^{(\ell)}_{K+1} )$.
	\begin{enumerate}
		\item Symmetric innovation:
	There exist some positive constants $C$ and $C_1$ such that
	\beq
		\label{eq:4.13}
		\left\| \rho(\widetilde\bV_K, \bV_K) \right\|_{\psi_1} \le C_1\sqrt{\frac{1}{m}\sum\limits_{\ell=1}^m S_{\ell}^2 }
	\eeq
	so long as $n\geq C \sqrt{K} \max_{\ell\in[m]}( \kappa_{\ell}^2 r_{\ell} ) $.		
	
		\item General distribution: There exist positive constant $C_2$ and $C_3$ such that when $n\geq \max_{\ell \in [m]} r_{\ell}$,
	\beq
		\label{eq:5.2}
		\left\| \rho(\widetilde\bV_K, \bV_K) \right\|_{\psi_1} \le C_2  \sqrt{\frac{1}{m}\sum\limits_{\ell=1}^m S_{\ell}^2}+ \frac{C_3}{m}\sum\limits_{\ell=1}^m B_{(\ell)} .
	\eeq
	\end{enumerate}

\end{thm}

\subsection{Heterogeneous factor models}
	
	Suppose on the server $\ell$, the data conform to a factor model as below.
	\[
		\bX_i^{(\ell)} = \bB^{(\ell)}\bff_i^{(\ell)} + \bu_i^{(\ell)}, ~~~i \in [n],
	\]
	where $\bB^{(\ell)}\in \RR^{d\times K}$ is the loading matrix, $\bff_i^{(\ell)}\in \RR^{K}$ is the factor that satisfies $\Cov (\bff_i^{(\ell)})= \bI$ and $\bu_i^{(\ell)}\in \RR^{d}$ is the residual vector. It is not hard to see that $\bSigma^{(\ell)}= \Cov (\bX^{(\ell)}_i)= \bB^{(\ell)}\bB^{(\ell)T} + \bSigma^{(\ell)}_u$, where $\bSigma^{(\ell)}_u$ is the covariance matrix of $\bu^{(\ell)}_i$.
	
	 Let $\bB^{(\ell)}\bB^{(\ell)T}= \bV_K^{(\ell)}\bLambda_K^{(\ell)}\bV_K^{(\ell)T}$ be the spectral decomposition of $\bB^{(\ell)}\bB^{(\ell)T}$. We assume that there exists a projection matrix $\bP_K=\bV_K\bV_K^T$, where $\bV_K\in\O_{d\times K}$, such that $\bV_K^{(\ell)}\bV_K^{(\ell)T}=\bP_K$ for all $\ell\in[m]$. In other words, $\{\bB^{(\ell)}\bB^{(\ell)T}\}_{\ell=1}^m$ share the same top $K$ eigenspace. Given the context of factor models, this implies that the factors have similar impact on the variation of the data across servers. Our goal now is to recover $\col(\bV_K)$ by the distributed PCA approach, namely Algorithm \ref{algo:1}.
	
	 Recall that $\hat{\bSigma}^{(\ell)}=\frac{1}{n}\sum_{i=1}^{n}\bX_i^{(\ell)}\bX_i^{(\ell)^T}$ is the sample covariance matrix on the $\ell$-th machine, and $\hat{\bV}_K^{(\ell)}=(\hat{\bv}^{(\ell)}_1,\cdots,\hat{\bv}^{(\ell)}_K)\in\O_{d\times K}$ stores $K$ leading eigenvectors of $\hat{\bSigma}^{(\ell)}$.
	 Define $\widetilde{\bSigma}=\frac{1}{m}\sum_{\ell=1}^{m}\hat{\bV}_K^{(\ell)}\hat{\bV}_K^{(\ell)T}$, and let $\widetilde{\bV}_K\in\O_{d\times K}$ be the top $K$ eigenvectors of $\widetilde{\bSigma}$.	
	 Below we present a theorem that characterizes the statistical performance of the distributed PCA under the heterogeneous factor models.

\begin{thm}
	\label{thm:7}
For each $\ell\in [m]$, let $S_{\ell}= \kappa_{\ell} \sqrt{\frac{K r_{\ell}}{N}}$ and $B_{\ell}=
\frac{\kappa^2_{\ell}\sqrt{K} r_{\ell} }{n}$.
There exist some positive constants $C_1$, $C_2$ and $C_3 $ such that when $n \geq \max_{\ell\in[m]} r_{\ell}$,
\begin{equation}
\begin{split}\label{eq:5.4}
		&\|\rho(\widetilde \bV_K, \bV_K)\|_{\psi_1} \le C_1
		\sqrt{\frac{1}{m}\sum\limits_{\ell=1}^m S_{\ell}^2}+\frac{C_2}{m}\sum\limits_{\ell=1}^m B_{\ell}
		+C_3\frac{\sqrt{K}}{m}\sum\limits_{\ell=1}^{m}\frac{\|\bSigma_u^{(\ell)}\|_{2}}{\lambda_K( \bLambda_K^{(\ell)} )}.
\end{split}
\end{equation}
\end{thm}
The first two terms in the RHS of \eqref{eq:5.4} are similar to those in \eqref{eq:5.2}, while the third term characterizes the effect of heterogeneity in statistical efficiency of $\widetilde\bV_K$. When $\|\bSigma^{(\ell)}_u\|_{2}$ is small compared with $\lambda_K( \bLambda_K^{(\ell)} )$ as in spiky factor models, $\bSigma^{(\ell)}_u$ can hardly distort the eigenspace $\col(\bV_K)$ and thus has little influence on the final statistical error of $\widetilde\bV_K$.

\section{Simulation study}
\label{sec:6}

In this section, we conduct Monte Carlo simulations to validate the statistical error rate of $\widetilde\bV_K$ that is established in the previous section. We also compare the statistical accuracy of $\widetilde\bV_K$ and its full sample counterpart $\hat\bV_K$, that is, the empirical top $K$ eigenspace based on the full sample covariance. The main message is that our proposed distributed estimator performs equally well as the full sample estimator $\hat\bV_K$ when the subsample size $n$ is large enough.

\subsection{Verification of the statistical error rate}

Consider $\{\bx_i\}_{i=1}^N$ i.i.d. following $N(\bzero, \bSigma)$, where $\bSigma= \diag(\lambda, \lambda/2, \lambda/4, 1, \cdots, 1)$. Here the number of spiky eigenvalues $K=3$ and $\bV_K=(\be_1, \be_2, \be_3)$. We generate $m$ subsamples, each of which has $n$ samples, and run our proposed distributed PCA algorithm (Algorithm \ref{algo:1}) to calculate $\widetilde\bV_K$. Since the centered multivariate Gaussian distribution is symmetric, according to Theorem \ref{thm:4}, when $\lambda=O(d)$ we have
\beq
	\label{eq:6.1}
	\|\rho(\widetilde\bV_K, \bV_K)\|_{\psi_1} =O \Bigl(\frac{C_1\opnorm{\bSigma}}{\lambda_K-\lambda_{K+1}}\sqrt{\frac{Kr(\bSigma)}{N}}\Bigr)= O\Bigl( \sqrt{\frac{d}{mn\delta}}\Bigr),
\eeq
where $\delta:=\lambda_K-\lambda_{K+1}= \lambda/4-1$. Now we provide numerical verification of the order of the number of servers $m$, the eigengap $\delta$, the subsample size $n$ and dimension $d$ in the statistical error.

\begin{figure}[H]
\centerline{
 \includegraphics[width=\textwidth]{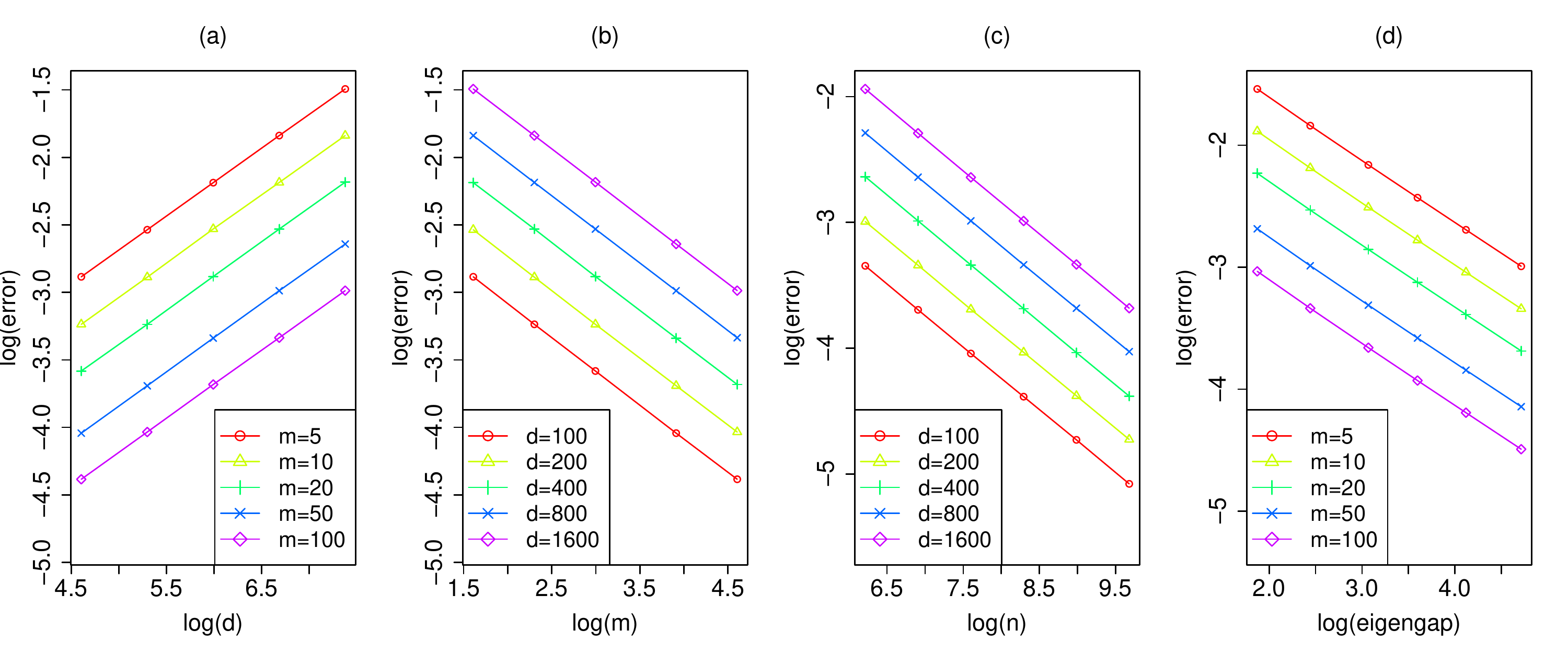}
}
\caption{Statistical error rate with respect to: (a) the dimension $d$ when $\lambda=50$ and $n=2000$; (b) the number of servers $m$ when $\lambda=50$ and $n=2000$; (c) the subsample size $n$ when $\lambda=50$ and $m=50$; (d) the eigengap $\delta$ when $d=800$ and $n=2000$.}
\label{fig:allrates}
\end{figure}

Figure \ref{fig:allrates} presents four plots that demonstrate how $\rho(\widetilde\bV_K, \bV_K)$ changes as $d$, $m$, $n$ and $\delta$ increases respectively. Each data point on the plots is based on $100$ independent Monte Carlo simulations. Figure \ref{fig:allrates}(a) demonstrates how $\rho(\widetilde\bV_K, \bV_K)$ increases with respect to the increasing dimension $d$ when $\lambda=50$ and $n=2000$. Each line on the plot represents a fixed number of machines $m$. Figure \ref{fig:allrates}(b) shows the decay rate of $\rho(\widetilde\bV_K, \bV_K)$ as the number of servers $m$ increases when $\lambda=50$ and $n=2000$. Different lines on the plot correspond to different dimensions $d$. Figure \ref{fig:allrates}(c) demonstrates how $\rho(\widetilde\bV_K, \bV_K)$ decays as the subsample size $n$ increases when $\lambda=50$ and $m=50$. Figure \ref{fig:allrates}(d) shows the relationship between $\rho(\widetilde\bV_K, \bV_K)$ and the eigengap $\delta$ when $d=800$ and $n=2000$. The results from Figures \ref{fig:allrates}(a)-\ref{fig:allrates}(d) show that $\rho(\widetilde\bV_K, \bV_K)$ is proportion to $d^{\frac{1}{2}}$, $m^{-\frac{1}{2}}$, $n^{-\frac{1}{2}}$ and $\delta^{-\frac{1}{2}}$ respectively when the other three parameters are fixed. These empirical results are all consistent with \eqref{eq:6.1}.

Figure \ref{fig:allrates} demonstrates the marginal relationship between $\rho(\widetilde\bV_K, \bV_K)$ and the four parameters $m$, $n$, $d$ and $\delta$. Now we study their joint relationship. Inspired by \eqref{eq:6.1}, we consider a multiple regression model as follows:
\begin{equation}
\label{eq:mregression}
\log(\rho(\widetilde\bV_K, \bV_K)) = \beta_0 + \beta_1 \log(d) + \beta_2 \log(m) +
\beta_3 \log(n) + \beta_4 \log(\delta) + \varepsilon ,
\end{equation}
where $\varepsilon$ is the error term. We collect all the data points $(d,m,n,\delta, \rho(\widetilde\bV_K, \bV_K))$ from four plots in Figure \ref{fig:allrates} to fit the regression model \eqref{eq:mregression}. The fitting result is that $\hat{\beta}_1 = 0.5043$, $\hat{\beta}_2 = -0.4995$, $\hat{\beta}_3 = -0.5011$ and $\hat{\beta}_4 = -0.5120$ with the multiple $R^2 = 0.99997$. These estimates are quite consistent with the theoretical results in \eqref{eq:6.1}. Moreover, Figure \ref{fig:fitted} plots all the observed values of $\log(\rho(\widetilde\bV_K, \bV_K))$ against its fitted values by the linear model \eqref{eq:mregression}. We can see that the observed and fitted values perfectly match. It indicates that the multiple regression model \eqref{eq:mregression} well explains the joint relationship between the statistical error and the four parameters $m$, $n$, $d$ and $\delta$.

\begin{figure}[]
\centerline{
 \includegraphics[width=.6\textwidth]{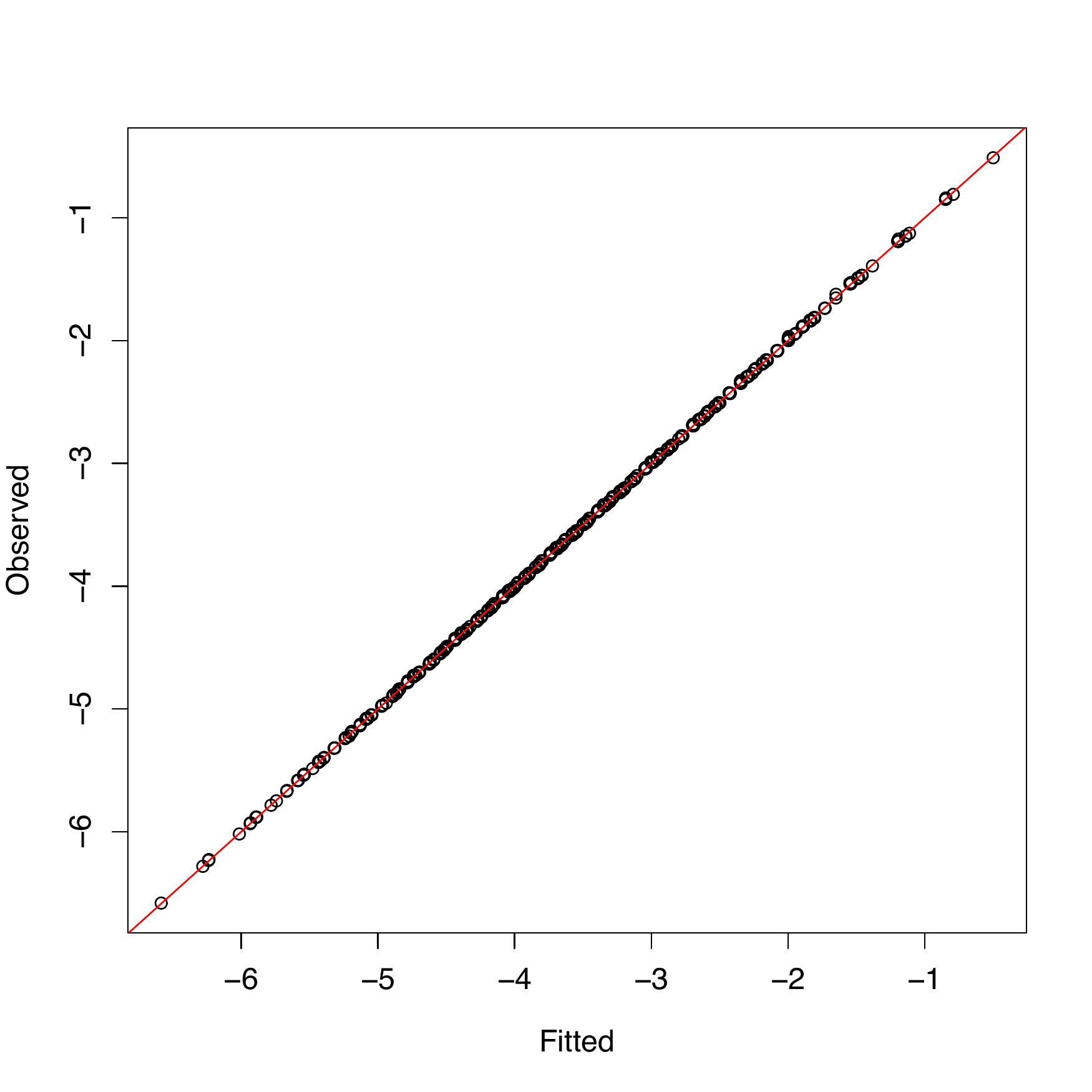}
}
\caption{Observed and fitted values of $\log(\rho(\widetilde\bV_K, \bV_K))$.}
\label{fig:fitted}
\end{figure}

\subsection{The effects of splitting}

In this section we investigate how the number of data splits $m$ affects the statistical performance of $\widetilde\bV_K$ when the total sample size $N$ is fixed. Since $N=mn$, it is easy to see that the larger $m$ is, the smaller $n$ will be, and hence the less computational load there will be on each individual server. In this way, to reduce the time consumption of the distributed algorithm, we prefer more splits of the data. However, per the assumptions of Theorem \ref{thm:4}, the subsample size $n$ should be large enough to achieve the optimal statistical performance of $\widetilde\bV_K$. This motivates us to numerically illustrate how $\rho(\widetilde\bV_K, \bV_K)$ changes as $m$ increases with $N$ fixed.

We adopt the same data generation process as described in the beginning of Section 6.1 with $\lambda=50$ and $N=6000$. We split the data into $m$ subsamples where $m$ is chosen to be all the factors of $N$ that are less than or equal to $300$. Figure \ref{fig:splitN} plots  $\rho(\widetilde\bV_K, \bV_K)$ with respect to the number of machines $m$. Each point on the plot is based on $100$ simulations. Each line corresponds to a different dimension $d$.

The results show that when the number of machines is not unreasonably large, or equivalently the number of subsample size $n$ is not small, the statistical error does not depend on the number of machines when $N$ is fixed. This is consistent with \eqref{eq:6.1} where the statistical error rate only depends on the total sample size $N=mn$.  When the number of machines $m$ is large ($\log m \ge 5$), or the subsample size $n$ is small, we observe slightly growing statistical error of the distributed PCA. This is aligned with the required lower bound of $n$ in Theorem \ref{thm:4} to achieve the optimal statistical performance of $\widetilde\bV_K$. Note that even when $m=300$ ($\log(m) \approx 5.7$) and $n=20$, our distributed PCA performs very well. This demonstrates that distributed PCA is statistically efficient as long as $m$ is within a reasonable range.

\begin{figure}[]
\centerline{
 \includegraphics[width=.8\textwidth]{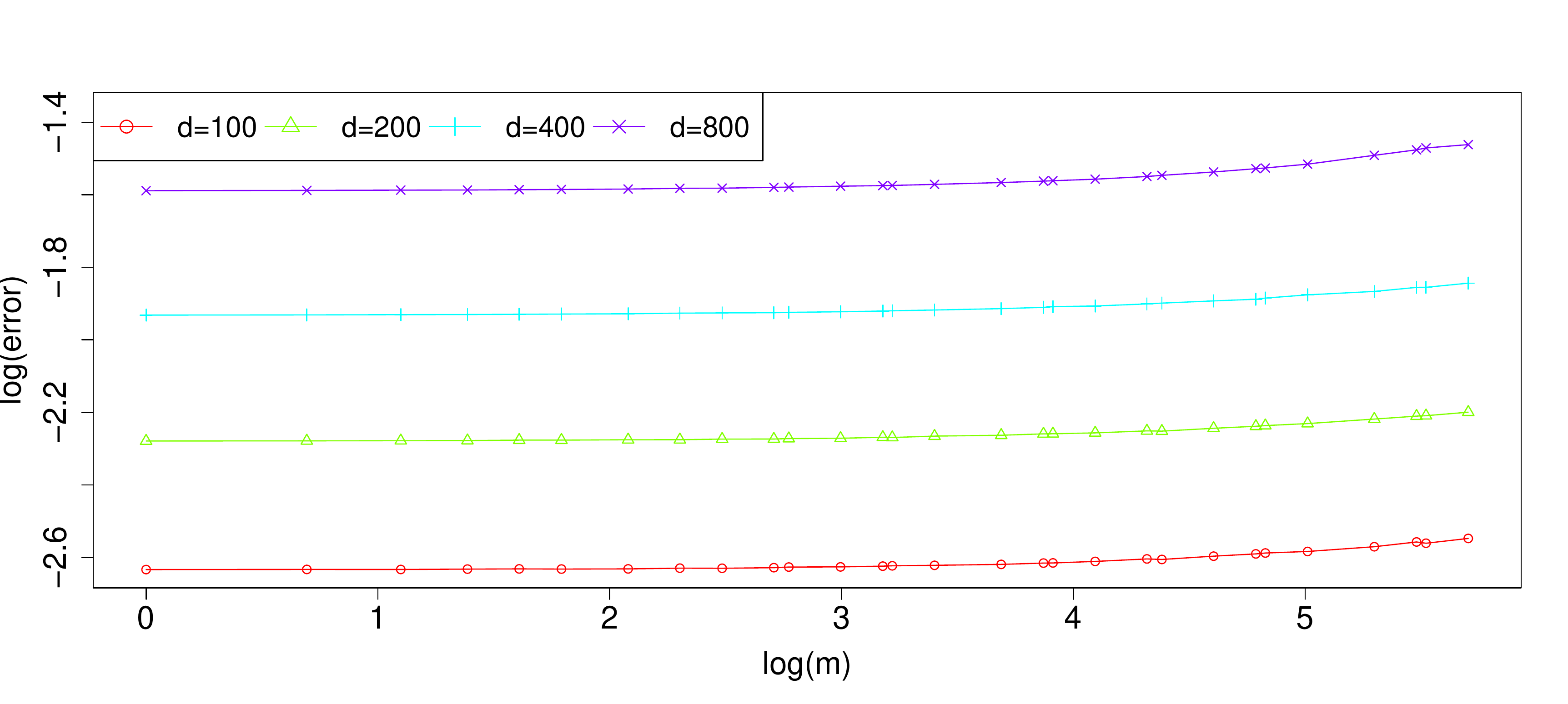}
}
\caption{Statistical error with respect to the number of machines when the total sample size $N=6000$ is fixed.}
\label{fig:splitN}
\end{figure}

\subsection{Comparison between distributed and full sample PCA}

In this subsection, we compare the statistical performance of the following three methods:
\begin{enumerate}
	\item Distributed PCA (DP)
	\item Full sample PCA (FP), i.e., the PCA based on the all the samples
	\item Distributed PCA with communication of five additional largest eigenvectors (DP5).
\end{enumerate}
Here we explain more on the third method DP5. The difference between DP5 and DP is that on each server, DP5 calculates $\hat\bV^{(\ell)}_{K+5}$, the top $K+5$ eigenvectors of $\bSigma^{(\ell)}$ and send them to the central server, and the central server computes the top $K$ eigenvectors of $(1/m)\sum_{\ell=1}^m \hat\bV^{(\ell)}_{K+5}\hat\bV_{K+5}^{(\ell)^T}$ as the final output. Intuitively, DP5 communicates more information of the covariance structure and is designed to guide the spill-over effects of the eigenspace spanned by the top $K$ eigenvalues. In Figure \ref{fig:comp}, we compare the performance of all the three methods under various scenarios.

\begin{figure}[H]
	\centering
	\includegraphics[width=.9\textwidth]{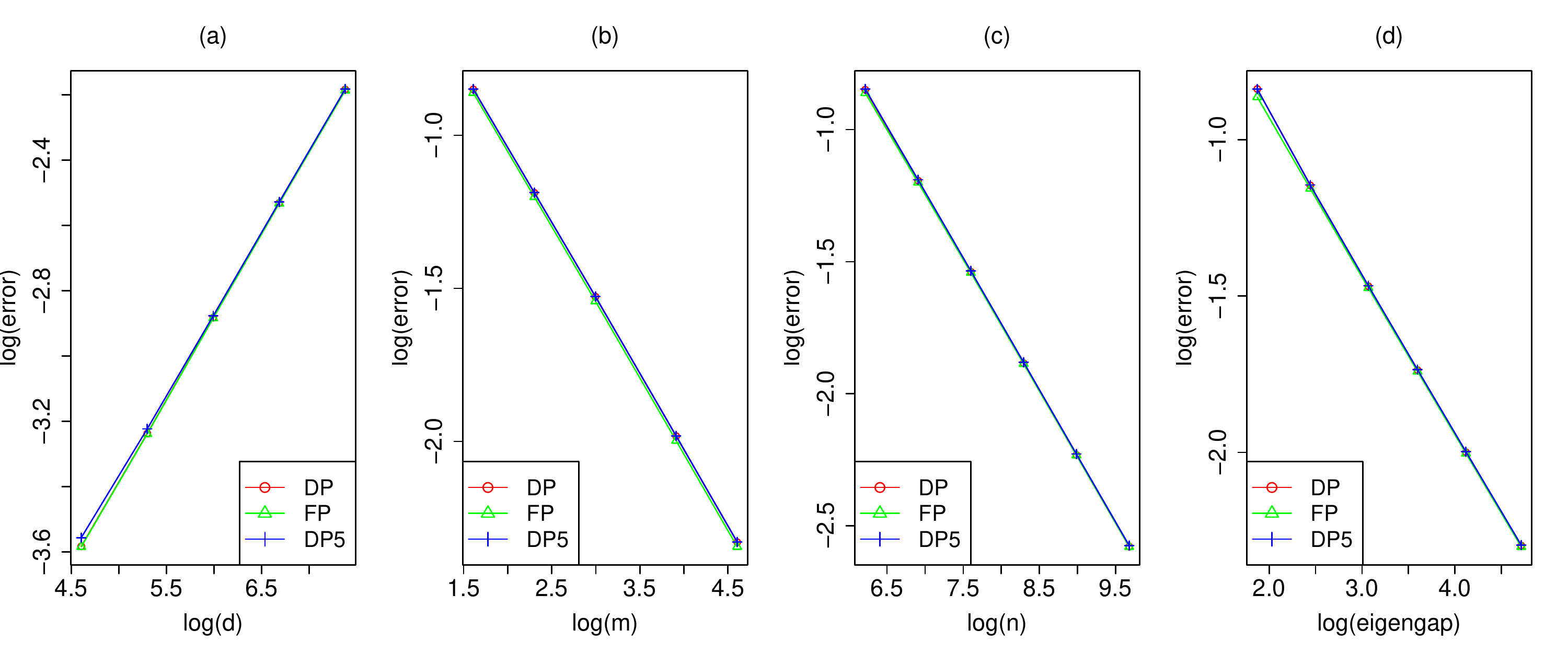}
	\caption{Comparison between DP, FP and DP5: (a) $m=20$, $n=2000$ and $\lambda=50$; (b) $d=1600$, $n=1000$ and $\lambda=30$; (c) $d=800$, $m=5$ and $\lambda=30$; (d) $d=1600$, $m=10$ and $n=500$.}
	\label{fig:comp}
\end{figure}

From Figures \ref{fig:comp}(a)-\ref{fig:comp}(d), we can see that all the three methods have similar finite sample performance. This means that it suffices to communicate $K$ eigenvectors to enjoy the same statistical accuracy as the full sample PCA. For more challenging situations with large $p/(mn\delta)$ ratios, small improvements using FP are visible.

\section{Discussion}

Our theoretical results are established under sub-Gaussian assumptions of the data.
We believe that similar results will hold under distributions with heavier tails than sub-Gaussian tails, or more specifically, with only bounded fourth moment. Typical examples are Student t-distributions with more than four degrees of freedom, Pareto distribution, etc. The only difference is that with heavy-tailed distribution, if the local estimators are still the top eigenspaces of the sample covariance matrix, we will not be able to derive exponential deviation bounds. To establish statistical rate with exponential deviation, special treatments of data, including shrinkage \citep{FWZ16,MSt16,WMi17},  are needed, and the bias induced by such treatments should be carefully controlled. This will be an interesting future problem to study.

\bibliographystyle{ims}
\bibliography{distributed}

\begin{thebibliography}{53}
\expandafter\ifx\csname natexlab\endcsname\relax\def\natexlab#1{#1}\fi
\expandafter\ifx\csname url\endcsname\relax
  \def\url#1{\texttt{#1}}\fi
\expandafter\ifx\csname urlprefix\endcsname\relax\def\urlprefix{URL }\fi

\bibitem[{Anderson(1963)}]{And63}
\textsc{Anderson, T.~W.} (1963).
\newblock Asymptotic theory for principal component analysis.
\newblock \textit{The Annals of Mathematical Statistics} \textbf{34} 122--148.

\bibitem[{Baik et~al.(2005)Baik, Ben~Arous and P{\'e}ch{\'e}}]{BBP05}
\textsc{Baik, J.}, \textsc{Ben~Arous, G.} and \textsc{P{\'e}ch{\'e}, S.}
  (2005).
\newblock Phase transition of the largest eigenvalue for nonnull complex sample
  covariance matrices.
\newblock \textit{The Annals of Probability} \textbf{33} 1643--1697.

\bibitem[{Battey et~al.(2015)Battey, Fan, Liu, Lu and Zhu}]{BFL15}
\textsc{Battey, H.}, \textsc{Fan, J.}, \textsc{Liu, H.}, \textsc{Lu, J.} and
  \textsc{Zhu, Z.} (2015).
\newblock Distributed estimation and inference with statistical guarantees.
\newblock \textit{arXiv preprint arXiv:1509.05457} .

\bibitem[{Bertrand and Moonen(2014)}]{BMo14}
\textsc{Bertrand, A.} and \textsc{Moonen, M.} (2014).
\newblock Distributed adaptive estimation of covariance matrix eigenvectors in
  wireless sensor networks with application to distributed {PCA}.
\newblock \textit{Signal Processing} \textbf{104} 120--135.

\bibitem[{Blanchard and M{\"u}cke(2017)}]{BMu17}
\textsc{Blanchard, G.} and \textsc{M{\"u}cke, N.} (2017).
\newblock Parallelizing spectral algorithms for kernel learning.
\newblock \textit{arXiv preprint arXiv:1610.07497} .

\bibitem[{Bosq(2000)}]{Bos00}
\textsc{Bosq, D.} (2000).
\newblock Stochastic processes and random variables in function spaces.
\newblock In \textit{Linear Processes in Function Spaces}. Springer, 15--42.

\bibitem[{Boutsidis et~al.(2016)Boutsidis, Woodruff and Zhong}]{BWZ15}
\textsc{Boutsidis, C.}, \textsc{Woodruff, D.~P.} and \textsc{Zhong, P.} (2016).
\newblock Optimal principal component analysis in distributed and streaming
  models.
\newblock In \textit{Proceedings of the 48th Annual ACM SIGACT Symposium on
  Theory of Computing}. ACM.

\bibitem[{Cai et~al.(2013)Cai, Ma and Wu}]{CMW13}
\textsc{Cai, T.~T.}, \textsc{Ma, Z.} and \textsc{Wu, Y.} (2013).
\newblock Sparse {PCA}: Optimal rates and adaptive estimation.
\newblock \textit{The Annals of Statistics} \textbf{41} 3074--3110.

\bibitem[{Chen et~al.(2016)Chen, Chang, Huang, Chen, Lin and Wang}]{CCH16}
\textsc{Chen, T.-L.}, \textsc{Chang, D.~D.}, \textsc{Huang, S.-Y.},
  \textsc{Chen, H.}, \textsc{Lin, C.} and \textsc{Wang, W.} (2016).
\newblock Integrating multiple random sketches for singular value
  decomposition.
\newblock \textit{arXiv preprint arXiv:1608.08285} .

\bibitem[{Chen and Xie(2014)}]{CXi14}
\textsc{Chen, X.} and \textsc{Xie, M.-g.} (2014).
\newblock A split-and-conquer approach for analysis of extraordinarily large
  data.
\newblock \textit{Statistica Sinica} \textbf{24} 1655--1684.

\bibitem[{Davis and Kahan(1970)}]{DKa70}
\textsc{Davis, C.} and \textsc{Kahan, W.~M.} (1970).
\newblock The rotation of eigenvectors by a perturbation. iii.
\newblock \textit{SIAM Journal on Numerical Analysis} \textbf{7} 1--46.

\bibitem[{El~Karoui and d'Aspremont(2010)}]{EdA10}
\textsc{El~Karoui, N.} and \textsc{d'Aspremont, A.} (2010).
\newblock Second order accurate distributed eigenvector computation for
  extremely large matrices.
\newblock \textit{Electronic Journal of Statistics} \textbf{4} 1345--1385.

\bibitem[{Fan et~al.(2016)Fan, Wang and Zhu}]{FWZ16}
\textsc{Fan, J.}, \textsc{Wang, W.} and \textsc{Zhu, Z.} (2016).
\newblock Robust low-rank matrix recovery.
\newblock \textit{arXiv preprint arXiv:1603.08315} .

\bibitem[{Garber et~al.(2017)Garber, Shamir and Srebro}]{GSS17}
\textsc{Garber, D.}, \textsc{Shamir, O.} and \textsc{Srebro, N.} (2017).
\newblock Communication-efficient algorithms for distributed stochastic
  principal component analysis.
\newblock \textit{arXiv preprint arXiv:1702.08169} .

\bibitem[{Golub and Van~Loan(2012)}]{GVa12}
\textsc{Golub, G.~H.} and \textsc{Van~Loan, C.~F.} (2012).
\newblock \textit{Matrix computations}, vol.~3.
\newblock JHU Press.

\bibitem[{Gross(2011)}]{Gro11}
\textsc{Gross, D.} (2011).
\newblock Recovering low-rank matrices from few coefficients in any basis.
\newblock \textit{IEEE Transactions on Information Theory} \textbf{57}
  1548--1566.

\bibitem[{Guo et~al.(2017)Guo, Lin and Zhou}]{GLZ17}
\textsc{Guo, Z.-C.}, \textsc{Lin, S.-B.} and \textsc{Zhou, D.-X.} (2017).
\newblock Learning theory of distributed spectral algorithms.
\newblock \textit{Inverse Problems} .

\bibitem[{Halko et~al.(2011)Halko, Martinsson and Tropp}]{HMT11}
\textsc{Halko, N.}, \textsc{Martinsson, P.-G.} and \textsc{Tropp, J.~A.}
  (2011).
\newblock Finding structure with randomness: Probabilistic algorithms for
  constructing approximate matrix decompositions.
\newblock \textit{SIAM review} \textbf{53} 217--288.

\bibitem[{Hotelling(1933)}]{Hot33}
\textsc{Hotelling, H.} (1933).
\newblock Analysis of a complex of statistical variables into principal
  components.
\newblock \textit{Journal of Educational Psychology} \textbf{24} 417--441.

\bibitem[{Johnstone(2001)}]{Joh01}
\textsc{Johnstone, I.~M.} (2001).
\newblock On the distribution of the largest eigenvalue in principal components
  analysis.
\newblock \textit{The Annals of Statistics} \textbf{29} 295--327.

\bibitem[{Johnstone and Lu(2009)}]{JLu12}
\textsc{Johnstone, I.~M.} and \textsc{Lu, A.~Y.} (2009).
\newblock On consistency and sparsity for principal components analysis in high
  dimensions.
\newblock \textit{Journal of the American Statistical Association} \textbf{104}
  682--693.

\bibitem[{Jung and Marron(2009)}]{JMa09}
\textsc{Jung, S.} and \textsc{Marron, J.~S.} (2009).
\newblock {PCA} consistency in high dimension, low sample size context.
\newblock \textit{The Annals of Statistics} \textbf{37} 4104--4130.

\bibitem[{Kannan et~al.(2014)Kannan, Vempala and Woodruff}]{KVW14}
\textsc{Kannan, R.}, \textsc{Vempala, S.} and \textsc{Woodruff, D.} (2014).
\newblock Principal component analysis and higher correlations for distributed
  data.
\newblock In \textit{Conference on Learning Theory}.

\bibitem[{Kargupta et~al.(2001)Kargupta, Huang, Sivakumar and Johnson}]{KHS01}
\textsc{Kargupta, H.}, \textsc{Huang, W.}, \textsc{Sivakumar, K.} and
  \textsc{Johnson, E.} (2001).
\newblock Distributed clustering using collective principal component analysis.
\newblock \textit{Knowledge and Information Systems} \textbf{3} 422--448.

\bibitem[{Kato(1966)}]{Kat66}
\textsc{Kato, T.} (1966).
\newblock \textit{Perturbation theory for linear operators}.
\newblock Springer.

\bibitem[{Kneip and Utikal(2001)}]{KUt01}
\textsc{Kneip, A.} and \textsc{Utikal, K.~J.} (2001).
\newblock Inference for density families using functional principal component
  analysis.
\newblock \textit{Journal of the American Statistical Association} \textbf{96}
  519--542.

\bibitem[{Koltchinskii and Lounici(2016)}]{KLo16}
\textsc{Koltchinskii, V.} and \textsc{Lounici, K.} (2016).
\newblock Asymptotics and concentration bounds for bilinear forms of spectral
  projectors of sample covariance.
\newblock In \textit{Annales de l'Institut Henri Poincar{\'e}, Probabilit{\'e}s
  et Statistiques}, vol.~52. Institut Henri Poincar{\'e}.

\bibitem[{Koltchinskii and Lounici(2017)}]{KLo17}
\textsc{Koltchinskii, V.} and \textsc{Lounici, K.} (2017).
\newblock Concentration inequalities and moment bounds for sample covariance
  operators.
\newblock \textit{Bernoulli} \textbf{23} 110--133.

\bibitem[{Lee et~al.(2017)Lee, Liu, Sun and Taylor}]{LSL15}
\textsc{Lee, J.~D.}, \textsc{Liu, Q.}, \textsc{Sun, Y.} and \textsc{Taylor,
  J.~E.} (2017).
\newblock Communication-efficient sparse regression.
\newblock \textit{Journal of Machine Learning Research} \textbf{18} 1--30.

\bibitem[{Li et~al.(2011)Li, Scaglione and Manton}]{LSM11}
\textsc{Li, L.}, \textsc{Scaglione, A.} and \textsc{Manton, J.~H.} (2011).
\newblock Distributed principal subspace estimation in wireless sensor
  networks.
\newblock \textit{IEEE Journal of Selected Topics in Signal Processing}
  \textbf{5} 725--738.

\bibitem[{Liang et~al.(2014)Liang, Balcan, Kanchanapally and Woodruff}]{LBK14}
\textsc{Liang, Y.}, \textsc{Balcan, M.-F.~F.}, \textsc{Kanchanapally, V.} and
  \textsc{Woodruff, D.} (2014).
\newblock Improved distributed principal component analysis.
\newblock In \textit{Advances in Neural Information Processing Systems}.

\bibitem[{Minsker(2016)}]{MSt16}
\textsc{Minsker, S.} (2016).
\newblock Sub-gaussian estimators of the mean of a random matrix with
  heavy-tailed entries.
\newblock \textit{arXiv preprint arXiv:1605.07129} .

\bibitem[{Nadler(2008)}]{Nad08}
\textsc{Nadler, B.} (2008).
\newblock Finite sample approximation results for principal component analysis:
  A matrix perturbation approach.
\newblock \textit{The Annals of Statistics} \textbf{36} 2791--2817.

\bibitem[{Onatski(2012)}]{Ona12}
\textsc{Onatski, A.} (2012).
\newblock Asymptotics of the principal components estimator of large factor
  models with weakly influential factors.
\newblock \textit{Journal of Econometrics} \textbf{168} 244--258.

\bibitem[{Paul(2007)}]{Pau07}
\textsc{Paul, D.} (2007).
\newblock Asymptotics of sample eigenstructure for a large dimensional spiked
  covariance model.
\newblock \textit{Statistica Sinica} \textbf{17} 1617--1642.

\bibitem[{Paul and Johnstone(2012)}]{PJo12}
\textsc{Paul, D.} and \textsc{Johnstone, I.~M.} (2012).
\newblock Augmented sparse principal component analysis for high dimensional
  data.
\newblock \textit{arXiv preprint arXiv:1202.1242} .

\bibitem[{Pearson(1901)}]{Pea01}
\textsc{Pearson, K.} (1901).
\newblock On lines and planes of closest fit to systems of point in space.
\newblock \textit{Philosophical Magazine Series 6} \textbf{2} 559--572.

\bibitem[{Qu et~al.(2002)Qu, Ostrouchov, Samatova and Geist}]{QOS02}
\textsc{Qu, Y.}, \textsc{Ostrouchov, G.}, \textsc{Samatova, N.} and
  \textsc{Geist, A.} (2002).
\newblock Principal component analysis for dimension reduction in massive
  distributed data sets.
\newblock In \textit{IEEE International Conference on Data Mining (ICDM)}.

\bibitem[{Reiss and Wahl(2016)}]{RWa16}
\textsc{Reiss, M.} and \textsc{Wahl, M.} (2016).
\newblock Non-asymptotic upper bounds for the reconstruction error of {PCA}.
\newblock \textit{arXiv preprint arXiv:1609.03779} .

\bibitem[{Schizas and Aduroja(2015)}]{SAd15}
\textsc{Schizas, I.~D.} and \textsc{Aduroja, A.} (2015).
\newblock A distributed framework for dimensionality reduction and denoising.
\newblock \textit{IEEE Transactions on Signal Processing} \textbf{63}
  6379--6394.

\bibitem[{Shen et~al.(2013)Shen, Shen and Marron}]{SSM13}
\textsc{Shen, D.}, \textsc{Shen, H.} and \textsc{Marron, J.~S.} (2013).
\newblock Consistency of sparse {PCA} in high dimension, low sample size
  contexts.
\newblock \textit{Journal of Multivariate Analysis} \textbf{115} 317--333.

\bibitem[{Shen et~al.(2016)Shen, Shen, Zhu and Marron}]{SSZ16}
\textsc{Shen, D.}, \textsc{Shen, H.}, \textsc{Zhu, H.} and \textsc{Marron, J.}
  (2016).
\newblock The statistics and mathematics of high dimension low sample size
  asymptotics.
\newblock \textit{Statistica Sinica} \textbf{26} 1747--1770.

\bibitem[{Stewart and Sun(1990)}]{Ste90}
\textsc{Stewart, G.~W.} and \textsc{Sun, J.} (1990).
\newblock \textit{Matrix perturbation theory}.
\newblock Academic Press.

\bibitem[{Tropp et~al.(2016)Tropp, Yurtsever, Udell and Cevher}]{TYU16}
\textsc{Tropp, J.~A.}, \textsc{Yurtsever, A.}, \textsc{Udell, M.} and
  \textsc{Cevher, V.} (2016).
\newblock Randomized single-view algorithms for low-rank matrix approximation.
\newblock \textit{arXiv preprint arXiv:1609.00048} .

\bibitem[{Vaccaro(1994)}]{Vac94}
\textsc{Vaccaro, R.~J.} (1994).
\newblock A second-order perturbation expansion for the svd.
\newblock \textit{SIAM Journal on Matrix Analysis and Applications} \textbf{15}
  661--671.

\bibitem[{Vershynin(2012)}]{Ver10}
\textsc{Vershynin, R.} (2012).
\newblock Introduction to the non-asymptotic analysis of random matrices.
\newblock \textit{Compressed Sensing, Theory and Applications}  210 -- 268.

\bibitem[{Vu et~al.(2013)Vu, Lei et~al.}]{VLe13}
\textsc{Vu, V.~Q.}, \textsc{Lei, J.} \textsc{et~al.} (2013).
\newblock Minimax sparse principal subspace estimation in high dimensions.
\newblock \textit{The Annals of Statistics} \textbf{41} 2905--2947.

\bibitem[{Wang(2015)}]{Wan15}
\textsc{Wang, R.} (2015).
\newblock Singular vector perturbation under gaussian noise.
\newblock \textit{SIAM Journal on Matrix Analysis and Applications} \textbf{36}
  158--177.

\bibitem[{Wang and Fan(2017)}]{WFa17}
\textsc{Wang, W.} and \textsc{Fan, J.} (2017).
\newblock Asymptotics of empirical eigen-structure for ultra-high dimensional
  spiked covariance model.
\newblock \textit{The Annals of Statistics} .

\bibitem[{Wei and Minsker(2017)}]{WMi17}
\textsc{Wei, X.} and \textsc{Minsker, S.} (2017).
\newblock Estimation of the covariance structure of heavy-tailed distributions.
\newblock In \textit{Advances in Neural Information Processing Systems}.

\bibitem[{Xu(2002)}]{Xu02}
\textsc{Xu, Z.} (2002).
\newblock Perturbation analysis for subspace decomposition with applications in
  subspace-based algorithms.
\newblock \textit{IEEE Transactions on Signal Processing} \textbf{50}
  2820--2830.

\bibitem[{Yu et~al.(2015)Yu, Wang and Samworth}]{YWS15}
\textsc{Yu, Y.}, \textsc{Wang, T.} and \textsc{Samworth, R.} (2015).
\newblock A useful variant of the {D}avis--{K}ahan theorem for statisticians.
\newblock \textit{Biometrika} \textbf{102} 315--323.

\bibitem[{Zhang et~al.(2013)Zhang, Duchi and Wainwright}]{ZDW13}
\textsc{Zhang, Y.}, \textsc{Duchi, J.~C.} and \textsc{Wainwright, M.~J.}
  (2013).
\newblock Divide and conquer kernel ridge regression.
\newblock In \textit{COLT}.

\end{thebibliography}
	
\section{Proofs and technical lemmas}

\subsection{Proof of main results}

\subsubsection{Proof of Lemma \ref{lem:1}}
\begin{proof}
It follows from concentration of sample covariance matrix (Lemma \ref{lem-cov-concentration}) that	$\left\| \| \hat{\bSigma}^{(1)} - \bSigma \|_2 \right\|_{\psi_1} \lesssim \lambda_1 \sqrt{r/n}$. By the variant of Davis-Kahan theorem in \cite{YWS15},
	\[
	\rho( \hat\bV^{(1)}_K , \bV_K )=
	\fnorm{\hat\bV^{(1)}_K\hat\bV^{(1)T}_K- \bV_K \bV_K^T }
	=\sqrt{2} \sin \bTheta ( \hat\bV^{(1)}_K, \bV_K )
	\lesssim \sqrt{K}\| \hat{\bSigma}^{(1)} - \bSigma \|_2/ \Delta
	.
	\]
	Hence
	\[
	\left\| \rho( \hat\bV^{(1)}_K , \bV_K ) \right\|_{\psi_1}
	\lesssim \sqrt{K}\left\| \| \hat{\bSigma}^{(1)} - \bSigma \|_2 \right\|_{\psi_1} / \Delta
	\lesssim \kappa \sqrt{ K r/n }.
	\]
	By Jensen's inequality,
	\begin{align*}
	&\| \bSigma^* - \bV_K \bV_K^T \|_F = \fnorm{ \E ( \hat\bV^{(1)}_K\hat\bV^{(1)^T}_K ) - \bV_K \bV_K^T }
	\leq \E \fnorm{ \hat\bV^{(1)}_K\hat\bV^{(1)^T}_K  - \bV_K \bV_K^T }\\
	&=\E \rho ( \hat\bV^{(1)}_K,\bV_K )
	\leq \left\| \rho( \hat\bV^{(1)}_K , \bV_K ) \right\|_{\psi_1}.
	\end{align*}
	Therefore,
	\begin{align*}
	\left\|
	\fnorm{\hat\bV^{(1)}_K\hat\bV^{(1)^T}_K- \bSigma^*}
	\right\|_{\psi_1}
	&\leq \left\|
	\fnorm{\hat\bV^{(1)}_K\hat\bV^{(1)^T}_K- \bV_K \bV_K^T }
	\right\|_{\psi_1} + \| \bSigma^* - \bV_K \bV_K^T \|_F\\
	&\leq 2 \left\| \rho( \hat\bV^{(1)}_K , \bV_K ) \right\|_{\psi_1}
	\lesssim \kappa \sqrt{\frac{K r}{n}}.
	\end{align*}

\end{proof}

\subsubsection{Proof of Theorem \ref{thm:1}}
\begin{proof}
When $\|\bSigma^*- \bV_K\bV_K^T\|_2<1/4$, the Weyl's inequality forces $\lambda_K(\bSigma^*)>\frac{3}{4}$ and $\lambda_{K+1}(\bSigma^*)<\frac{1}{4}$. The Theorem 2 in \cite{YWS15} yields
\beq
\label{eq:6.3}
\rho(\widetilde\bV_K, \bV_K^*)
=\sqrt{2} \sin\bTheta (\widetilde\bV_K, \bV_K^*)
\lesssim  \frac{ \fnorm{\widetilde\bSigma- \bSigma^*} }{ \lambda_K(\bSigma^*)-\lambda_{K+1}(\bSigma^*)} \lesssim \fnorm{\widetilde\bSigma- \bSigma^*}.
\eeq
When $n\geq r$, Lemma \ref{lem-Bosq} and Lemma \ref{lem:1} imply that
\begin{align*}
&\left\| \|\widetilde{\bSigma}-\bSigma^*\|_F \right\|_{\psi_1}=
\left\| \left\|\frac{1}{m}\sum_{\ell=1}^m \hat{\bV}^{(\ell)}_K\hat{\bV}^{(\ell)T}_K- \bSigma^* \right\|_F \right\|_{\psi_1}
\lesssim \frac{1}{\sqrt{m}} \left\| \fnorm{\hat\bV^{(1)}_K\hat\bV^{(1)T}_K- \bSigma^*} \right\|_{\psi_1}
\lesssim \kappa \sqrt{\frac{Kr}{N}}.
\end{align*}
Combining the two inequalities above finishes the proof.

\end{proof}

\subsubsection{Proof of Theorem \ref{thm:2}}
\begin{proof}
	Choose $j \in[d]$ and let $\bD_j=\bI-2\be_j \be_j^T$.  Let $\bSigma=\bV\bLambda\bV^T$ be the spectral decomposition of $\bSigma$.  Assume that $\hat\lambda$ is an eigenvalue of the sample covariance $\hat\bSigma= (1/n)\sum\limits_{i=1}^n \bX_i\bX_i^T$ and $\hat\bv\in\S^{d-1}$ is the correspondent eigenvector that satisfies $\hat\bSigma\hat\bv=\hat \lambda\hat\bv$.
	
	Define $\bZ_i= \bLambda^{-\frac{1}{2}}\bV^T\bX_i$ and $\hat\bS= (1/n)\sum\limits_{i=1}^n \bZ_i\bZ_i^T$. Note that $\hat\bSigma= \bV\bLambda^{\frac{1}{2}}\hat\bS\bLambda^{\frac{1}{2}}\bV^T$. Consider the matrix $\check\bSigma= \bV\bLambda^{\frac{1}{2}}\bD_j\hat\bS\bD_j\bLambda^{\frac{1}{2}}\bV^T$. By the sign symmetry, $\hat\bSigma$ and $\check\bSigma$ are identically distributed. It is not hard to verify that $\check\bSigma$ also has an eigenvalue $\hat \lambda$ with the correspondent eigenvector being $\bV\bD_j \bV^T\hat\bv$. Denote the top $K$ eigenvectors of $\hat\bSigma$ by $\hat\bV_K= (\hat\bv_1, \cdots, \hat\bv_K)$ and the top $K$ eigenvectors of $\check\bSigma$ by $\check\bV_K$. Therefore we have
	\[
	\begin{aligned}
	\bV^T\E (\hat\bV_K\hat\bV_K^T)\bV & = \bV^T\E (\check\bV_K\check\bV_K^T)\bV= \bV^T\bV\bD_j\bV^T\E( \hat\bV_K\hat\bV_K^T) \bV\bD_j\bV^T\bV \\
	& = \bD_j\bV^T\E (\hat\bV_K\hat\bV_K^T) \bV\bD_j.
	\end{aligned}
	\]
	Since the equation above holds for all $j \in [d]$, we can reach the conclusion that $\bV^T \E (\hat\bV_K\hat\bV_K^T) \bV$ is diagonal, i.e, $ \bSigma^* \E (\hat\bV_K\hat\bV_K^T)$ and $\bSigma$ share the same set of eigenvectors.
	
	Suppose that $\| \bSigma^*- \bV_K\bV_K^T \|_2 <1/2$. As demonstrated above, for any $k \in [K]$, the $k$th column of $\bV_K$, which we denote by $\bv_k$, should be an eigenvector of $\bSigma^*$. Note that
	\[
	\begin{aligned}
	\ltwonorm{\bSigma^*\bv_k} & =\ltwonorm{(\bSigma^*- \bV_K\bV_K^T + \bV_K\bV_K^T)\bv_k} \ge 1 -\ltwonorm{
		\bSigma^* -  \bV_K\bV_K^T} > 1-\frac{1}{2}= \frac{1}{2}.
	\end{aligned}
	\]
	With regard to $\bSigma^*$, the correspondent eigenvalue of $\bv_k$ must be greater than $1/2$. Denote any eigenvector of $\bSigma$ that is not in $\{\bv_k\}_{k=1}^K$ by $\bu$, then analogously,
	\[
	\begin{aligned}
	\ltwonorm{ \bSigma^* \bu} & = \ltwonorm{(\bSigma^*- \bV_K\bV_K^T + \bV_K\bV_K^T)\bu} \le \ltwonorm{\bSigma^*- \bV_K\bV_K^T} < \frac{1}{2}.
	\end{aligned}
	\]
	For $\bSigma^*$, the correspondent eigenvalue of $\bu$ is smaller than $1/2$. Therefore, the top $K$ eigenspace of $\bSigma^*$ is exactly $\mathrm{Col}(\bV_K)$, and $\rho(\bV^*_K, \bV_K)=0 $.
\end{proof}

\subsubsection{Proof of Lemma \ref{lem:2}}
\begin{proof}
Note that $\|f(\cdot)\|_F \leq \Delta^{-1} \| \cdot \|_F$ and
\[
\| f(\bE\bU) \|_F \leq \Delta^{-1} \| \bE \bU\|_F \leq \Delta^{-1} \sqrt{K}\| \bE \bU\|_2 \leq
\Delta^{-1} \sqrt{K} \| \bE \|_2 = \sqrt{K} \varepsilon.
\]
Hence Lemma \ref{lem:2} is a direct corollary of Lemma \ref{lem:DK-strong}.
\end{proof}

\subsubsection{Proof of Theorem \ref{thm:3}}
\begin{proof}
	Define $\bE = \hat{\bSigma}^{(1)} - \bSigma$, $\bP=\bV_K \bV_K^T$, $\hat{\bP} = \hat{\bV}_K^{(1)} \hat{\bV}_K^{(1)T}$, $\bQ = f(\bE \bV_K) \bV_K^T + \bV_K f(\bE\bV_K)^T$, $\bW = \hat{\bP} - \bP - \bQ$ and $\varepsilon = \| \bE \|_2 / \Delta$. From $\E\bQ = \mathbf{0}$ and
	\begin{align*}
	& \hat{\bP} - \bP - \bQ = \bW = \bW 1_{ \{\varepsilon \leq 1/10\} } + ( \bW + \bQ ) 1_{ \{\varepsilon > 1/10\} } - \bQ 1_{ \{\varepsilon > 1/10\} }\notag\\
	&=\bW 1_{ \{\varepsilon \leq  1/10\} } + ( \hat{\bP} - \bP) 1_{ \{\varepsilon > 1/10\} } - \bQ 1_{ \{\varepsilon > 1/10\} },
	\end{align*}
	we derive that
	\begin{align}
	& \E \hat{\bP} - \bP = \E ( \bW 1_{ \{\varepsilon \leq 1/10\} } )
	+\E [( \hat{\bP} - \bP) 1_{ \{\varepsilon > 1/10\} }] - \E (\bQ 1_{ \{\varepsilon > 1/10\} }),\notag\\
	&\| \E \hat{\bP} - \bP \|_F \leq \E ( \| \bW \|_F 1_{ \{\varepsilon \leq 1/10\} } )
	+\E ( \| \hat{\bP} - \bP \|_F 1_{ \{\varepsilon > 1/10\} } ) + \E (\|\bQ\|_F 1_{ \{\varepsilon > 1/10\} }).\label{eqn-bias-1}
	\end{align}
	We are going to bound the three terms separately. On the one hand, Lemma \ref{lem:2} implies that $\| \bW\|_F \leq 24 \sqrt{K} \varepsilon^2$ when $\varepsilon\leq 1/10$. Hence
	\begin{align}
	\E ( \| \bW \|_F 1_{ \{\varepsilon \leq 1/10\} } )
	\leq   \E  ( 24\sqrt{K} \varepsilon^2 1_{ \{\varepsilon \leq 1/10\} } )
	\lesssim \sqrt{K} \E \varepsilon^2
	.\label{eqn-bias-2}
	\end{align}
	On the other hand, the Davis-Kahan theorem shows that $ \| \hat{\bP} - \bP \|_F\lesssim \sqrt{K}\varepsilon$. Besides, it is easily seen that $\| \bQ \|_F \lesssim \|f(\bE\bV_K) \|_F \leq \sqrt{K} \| \bE \|_2/\Delta =\sqrt{K} \varepsilon$. Hence
	\begin{align}
&	\E (\| \hat{\bP} - \bP \|_F 1_{ \{\varepsilon > 1/10\} } ) + \E (\|\bQ\|_F 1_{ \{\varepsilon > 1/10\} })
	 \lesssim \sqrt{K} \E ( \varepsilon 1_{ \{\varepsilon > 1/10\} } ) \notag\\
	& \leq 10 \sqrt{K} \E (\varepsilon^2 1_{ \{\varepsilon > 1/10\} } ) \lesssim \sqrt{K} \E \varepsilon^2
	.\label{eqn-bias-4}
	\end{align}
By (\ref{eqn-bias-1}),(\ref{eqn-bias-2}), (\ref{eqn-bias-4}) and Lemma \ref{lem-cov-concentration} we have
	\begin{align}
	&\| \E \hat{\bP} - \bP \|_F \lesssim \sqrt{K} \E \varepsilon^2
	= \sqrt{K}\Delta^{-2} \E \|\bE \|_2^2 \lesssim \sqrt{K}\Delta^{-2} \left\| \|\bE \|_2 \right\|_{\psi_1}^2 \lesssim \frac{\kappa^2 \sqrt{K} r}{n}.\label{eqn-bias-6}
	\end{align}
\end{proof}

\subsubsection{Proof of Theorem \ref{thm:4}}

\begin{proof}
According to Theorem \ref{thm:3}, there exists a constant $C$ such that $\| \bSigma^* - \bV_K \bV_K^T \|_2 \le 1/4$ as long as $n \geq C \kappa^2 \sqrt{K} r \geq r$. Then Theorem \ref{thm:1} implies that $\left\| \rho(\widetilde\bV_K, \bV_K^*) \right\|_{\psi_1}
\leq C_1 \kappa \sqrt{\frac{Kr}{N}}$ for some constant $C_1$.

When random samples have symmetric innovation, we have $\rho(\bV^*_K, \bV_K)=0$ and
\[
\left\| \rho(\widetilde\bV_K, \bV_K) \right\|_{\psi_1} =
\left\| \rho(\widetilde\bV_K, \bV_K^*) \right\|_{\psi_1} \leq C_1 \kappa \sqrt{\frac{Kr}{N}}.
\]
	
For general distribution, Theorem \ref{thm:3} implies that $\rho(\bV^*_K, \bV_K)
\leq C_2 \kappa^2 \sqrt{K} r / n$ for some constant $C_2$ and
\begin{align}
\left\| \rho(\widetilde\bV_K, \bV_K) \right\|_{\psi_1} \leq
\left\| \rho(\widetilde\bV_K, \bV_K^*) \right\|_{\psi_1} + \rho(\bV^*_K, \bV_K) \leq C_1 \kappa \sqrt{\frac{Kr}{N}} + C_2 \kappa^2 \frac{\sqrt{K} r}{n}.
\label{ineq-dpca-general-proof}
\end{align}
When $m\leq C_3 n/(\kappa^2 r)  $ for some constant $C_3$, we have
\[
\kappa \sqrt{\frac{Kr}{N}} =  \sqrt{\frac{\kappa^2 Kr}{nm}} \geq  \sqrt{\frac{\kappa^2 Kr}{n\cdot C_3 n/(\kappa^2 r)}} = \frac{1}{\sqrt{C_3}}\cdot \frac{\kappa^2  \sqrt{K} r}{n},
\]
and \eqref{ineq-dpca-general-proof} forces
\[
\left\| \rho(\widetilde\bV_K, \bV_K) \right\|_{\psi_1} \leq (C_1 + C_2 \sqrt{C_3}) \kappa \sqrt{\frac{Kr}{N}}.
\]
\end{proof}

\subsubsection{Proof of Theorem \ref{thm:5}}
\begin{proof}
We first focus on the first subsample $\{ \bX_i^{(1)} \}_{i=1}^n$ and the associated top eigenvector $\hat\bV^{(1)}_1$. For ease of notation, we temporarily drop the superscript. Let $S=\sum_{i=1}^{n}W_i$ and $\widehat{\bSigma}_{\bZ}=\frac{d-1}{n}\sum_{i=1}^{n}1_{\{W_i=1\}}\bZ_i\bZ_i^T$. From $\widehat{\bSigma}=
\begin{pmatrix}
\frac{2\lambda}{n}(n-S) & \mathbf{0}_{1\times(d-1)}\\
\mathbf{0}_{(d-1)\times 1} & 2\hat\bSigma_{\bZ}
\end{pmatrix}
$ we know that $\|\widehat{\bSigma}_{\bZ}\|_{2}>(\lambda/n)(n-S)$ and $\|\hat{\bSigma}_{\bZ}\|_{2}<(\lambda/n)(n-S)$ lead to $\hat{\bV}_1\perp \bV_1$ and $\hat{\bV}_1\sslash \bV_1$ (i.e. $\hat{\bV}_1 = \pm \bV_1$), respectively. Besides, $\|\widehat{\bSigma}_{\bZ}\|_{2}$ is a continuous random variable. Hence $\P(\hat{\bV}_1 \perp \bV_1)+\P(\hat{\bV}_1 \sslash \bV_1)=1$. Note that
\begin{equation*}
\begin{split}
&\Tr(\hat\bSigma_{\bZ})=\frac{d-1}{n}\sum_{i=1}^{n}1_{\{W_i=1\}}=\frac{(d-1)S}{n},\\
&\|\hat\bSigma_{\bZ}\|_{2}\geq\frac{\Tr(\hat\bSigma_{\bZ})}{\rank(\hat\bSigma_{\bZ})}
\geq\frac{\Tr(\hat\bSigma_{\bZ})}{\min\{n,d-1\}}
\geq\frac{(d-1)S}{n^2}.
\end{split}
\end{equation*}
Then
\begin{equation*}
\begin{split}
&\P(\hat{\bV}_1\sslash \bV_1)\leq
\P\Big(\|\hat\bSigma_{\bZ}\|_{2}\leq\frac{\lambda}{n}(n-S)\Big)
\leq \P\Big(\frac{(d-1)S}{n^2}\leq\frac{\lambda}{n}(n-S)\Big)\\
&=\P\Big(\frac{S}{n}\leq\frac{1}{1+\frac{d-1}{n\lambda}}\Big)
\leq \P\Big(\frac{S}{n}\leq\frac{1}{4}\Big)=\P\Big(\frac{S}{n}-\frac{1}{2}\leq -\frac{1}{4}\Big)
\leq e^{-n/8}.
\end{split}
\end{equation*}
Above we used the assumption $d\geq 3n\lambda+1$ and Hoeffding's inequality. Now we finish the analysis of $\hat\bV^{(1)}_1$ and collect back the superscript.

From now on we define $S=\sum_{\ell=1}^{m}1_{\{\hat{\bV}^{(\ell)}_1 \sslash \bV_1\}}$. For $\hat{\bV}^{(\ell)}_1$, let $a_{\ell}$ be its first entry and $\bb_{\ell}$ be the vector of its last $(d-1)$ entries. The dichotomy $\P(\hat{\bV}^{(\ell)}_1\sslash \bV_1)+\P(\hat{\bV}^{(\ell)}_1\perp \bV_1)=1$ mentioned above forces $|a_{\ell}|=1_{\{ \hat{\bV}^{(\ell)}_1 \sslash \bV_1\}}$, $\|\bb_{\ell}\|_2=1_{\{\hat{\bV}^{(\ell)}_1\perp \bV_1\}}$, $\hat{\bV}^{(\ell)}_1 \hat{\bV}^{(\ell)T}_1=\begin{pmatrix}
1_{\{ \hat{\bV}^{(\ell)}_1 \sslash \bV_1\}} & \mathbf{0}_{1\times(d-1)}\\
\mathbf{0}_{(d-1)\times 1} & \bb_{\ell} \bb_{\ell}^T
\end{pmatrix}$, and
\begin{equation*}
\widetilde{\bSigma}=\frac{1}{m}\sum_{\ell =1}^{m} \hat{\bV}^{(\ell)}_1 \hat{\bV}^{(\ell)T}_1=\begin{pmatrix}
\frac{1}{m}S& \mathbf{0}_{1\times(d-1)}\\
\mathbf{0}_{(d-1)\times 1} & \frac{1}{m}\sum_{\ell=1}^{m}\bb_{\ell}\bb_{\ell}^T
\end{pmatrix}
.
\end{equation*}
Note that $n\geq 32\log d$ forces $\P( \hat{\bV}_1^{(\ell)}\sslash \bV_1)\leq e^{-n/8}\leq d^{-4}$.

\textbf{Case 1: $m\leq d^3$}

In this case, $\P(S=0)=[1-\P(\hat{\bV}_1^{(1)}\sslash \bV_1)]^m\geq 1-m \P( \hat{\bV}_1^{(1)}\sslash \bV_1)\geq 1-d^{-1}$. When $S=0$, we have $\|\bb_{\ell}\|_2=1$ for all $\ell \in[m]$ and $\|(1/m)\sum_{\ell=1}^{m}\bb_{\ell}\bb_{\ell}^T\|_{2}>0$, leading to $\widetilde{\bV}_1\perp \bV_1$.

\textbf{Case 2: $m> d^3$}

On the one hand, by Hoeffding's inequality we obtain
\begin{equation*}
\P\left(\frac{S}{m}\geq \frac{1}{d}\right)\leq \P\left(\frac{1}{m}(S-\E S)\geq \frac{1}{2d}\right)\leq e^{-2m(\frac{1}{2d})^2}=e^{-\frac{m}{2d^2}}<e^{-d/2}.
\end{equation*}

On the other hand, note that
\begin{equation*}
\begin{split}
& \left\|\frac{1}{m}\sum_{k=1}^{m}\bb_{\ell}\bb_{\ell}^T
\right\|_{2} \geq \frac{ \Tr\left( \frac{1}{m}\sum_{k=1}^{m}\bb_{\ell}\bb_{\ell}^T \right) }{d-1}
=\frac{\frac{1}{m}\sum_{k=1}^{m}\|\bb_{\ell}\|_2^2}{d-1}
=\frac{1}{d-1}\Big(1-\frac{S}{m}\Big).
\end{split}
\end{equation*}
Hence
\begin{equation*}
\P(\widetilde{\bV}_1\perp \bV_1)\geq \P\Big(\Big\|\frac{1}{m}\sum_{k=1}^{m}\bb_{\ell}\bb_{\ell}^T\Big\|_{2}>\frac{S}{m}\Big)
\geq \P\Big[\frac{1}{d-1}\Big(1-\frac{S}{m}\Big)>\frac{S}{m}\Big]=\P\Big(\frac{S}{m}<\frac{1}{d}\Big)\geq 1-e^{-d/2}.
\end{equation*}

\end{proof}

	\subsubsection{Proof of Theorem \ref{thm:6}}
	\begin{proof}
		With slight abuse of notations, here we define $\bSigma^*_{\ell}=\E ( \hat\bV^{(\ell)}_K\hat\bV_K^{(\ell)T} )$,
		$\bSigma^*=\frac{1}{m}\sum\limits_{\ell=1}^m  \bSigma^*_{\ell}$, and $\bV_K^*\in \RR^{d\times K}$ to be the top $K$ eigenvectors of $\bSigma^*$.

		First we consider the general case. Note that $\lambda_{K}(\bV_K \bV_K^T)=1$ and $\lambda_{K}(\bV_K \bV_K^T)=0$. By the Davis-Kahan theorem, we have
		\begin{align}
		\rho( \widetilde{\bV}_K , \bV_K ) \lesssim \| \widetilde{\bSigma} - \bV_K \bV_K^T \|_F \leq
		\| \widetilde{\bSigma} - \bSigma^* \|_F + \| \bSigma^* - \bV_K \bV_K^T \|_F.
		\label{ineq-common-principal}
		\end{align}
		Note that $\bSigma^* = \frac{1}{m} \sum_{\ell=1}^{m} \bSigma^{*}_{\ell}$. The first term in \eqref{ineq-common-principal} is the norm of independent sums
		\begin{align*}
		\| \widetilde{\bSigma} - \bSigma^* \|_F = \left\| \frac{1}{m} \sum_{\ell=1}^{m} \left( \hat{\bV}_K^{(\ell)} \hat{\bV}_K^{(\ell)T} - \bSigma^{*}_{\ell} \right) \right\|_F
		\end{align*}
		It follows from Lemma \ref{lem:1} that $\left\| \| \hat{\bV}_K^{(\ell)} \hat{\bV}_K^{(\ell)T} - \bSigma^{*}_{\ell} \|_F \right\|_{\psi_1}
		\lesssim \kappa_{\ell} \sqrt{\frac{K r_{\ell}}{n}}=\sqrt{m}S_{\ell}$, from which Lemma \ref{lem-Bosq} leads to
		\begin{align}
		&\left\|  \| \widetilde{\bSigma} - \bSigma^* \|_F \right\|_{\psi_1} \lesssim \frac{1}{m}
		\sqrt{\sum_{\ell=1}^{m} \left( \sqrt{m}S_{\ell} \right)^2 }
		=\sqrt{\frac{1}{m} \sum_{\ell=1}^{m} S_{\ell}^2 }.
		\label{ineq-common-principal-1}
		\end{align}
		The second term in \eqref{ineq-common-principal} is bounded by
		\begin{align*}
		\| \bSigma^* - \bV_K \bV_K^T \|_F = \left\| \frac{1}{m} \sum_{\ell=1}^{m} \left( \bSigma^{*}_{\ell} - \bV_K \bV_K^T \right) \right\|_F \leq \frac{1}{m} \sum_{\ell=1}^{m} \left\| \bSigma^{*}_{\ell} - \bV_K \bV_K^T \right\|_F.
		\end{align*}
		Theorem \ref{thm:3} implies that when $n\geq r_{\ell}$,
		\begin{align}
		\left\| \bSigma^{*}_{\ell} - \bV_K \bV_K^T \right\|_F \lesssim \kappa^2_{\ell} \sqrt{K} r_{\ell}/n = B_{\ell}.
		\label{ineq-common-principal-1.5}
		\end{align}
		Hence
		\begin{align}
		\| \bSigma^* - \bV_K \bV_K^T \|_F \lesssim \frac{1}{m} \sum_{\ell=1}^{m} B_{\ell}.
		\label{ineq-common-principal-2}
		\end{align}
		The claim under general case follows from \eqref{ineq-common-principal}, \eqref{ineq-common-principal-1} and \eqref{ineq-common-principal-2}.
		
		Now we come to the symmetric case. If $\| \bSigma^*_{\ell} - \bV_K \bV_K^T \|_{2} < 1/2$ for all $\ell\in [m]$, then Theorem \ref{thm:2} implies that the top $K$ eigenspace of $\bSigma^*_{\ell}$ is $\col(\bV_K)$. Therefore, the top $K$ eigenspace of $\bSigma^*$ is still $\col(\bV_K)$ and  $\rho(\bV_K, \bV_K^*)=0$.
		
		When $n\geq C \sqrt{K} \max_{\ell\in[m]}( \kappa_{\ell}^2 r_{\ell} ) $ for large $C$, \eqref{ineq-common-principal-1.5} ensures $\max_{\ell\in[m]}\| \bSigma^*_{\ell} - \bV_K \bV_K^T \|_{2} \leq 1/4$, $\| \bSigma^* - \bV_K \bV_K^T \|_{2} \leq 1/4$ and $\rho(\bV_K, \bV_K^*)=0$. Weyl's inequality forces $\lambda_K( \bSigma^* ) \geq 3/4$ and $\lambda_{K+1}( \bSigma^* ) \leq 1/4$. By the Davis-Kahan theorem and \eqref{ineq-common-principal-1},
		\[
		\left\| \rho( \widetilde{\bV}_K , \bV_K ) \right\|_{\psi_1}
		=\left\| \rho( \widetilde{\bV}_K , \bV^*_K ) \right\|_{\psi_1}
		\lesssim \left\| \| \widetilde{\bSigma} - \bSigma^* \|_F \right\|_{\psi_1} \lesssim \sqrt{ \frac{1}{m}\sum_{\ell=1}^{m} S_{\ell}^2 }.
		\]
	\end{proof}

\subsubsection{Proof of Theorem \ref{thm:7}}
\begin{proof}
We define $\bSigma^*_{\ell}= \E ( \hat{\bV}_K^{(\ell)}\hat{\bV}_K^{(\ell)T})$ and $\bSigma^*= \frac{1}{m} \sum_{\ell=1}^{m} \bSigma^*_{\ell}$. Let $\bV_K^*,\bar{\bV}_K^{(\ell)}\in\O_{d\times K}$ be the top $K$ eigenvectors of $\bSigma^*$ and $\bSigma^{(\ell)}$, respectively.
By the Davis-Kahan theorem,
\begin{align}
&\rho( \widetilde{\bV}_K , \bV_K ) \lesssim \| \widetilde{\bSigma} - \bV_K \bV_K^T \|_F \leq
\| \widetilde{\bSigma} - \bSigma^* \|_F + \| \bSigma^* - \bV_K \bV_K^T \|_F
.\label{ineq-factor-1}
\end{align}
The first term in \eqref{ineq-factor-1} is controlled in exactly the same way as \eqref{ineq-common-principal-1}. The second term is further decomposed as
\begin{align}
& \| \bSigma^* - \bV_K \bV_K^T \|_F
= \left\| \frac{1}{m} \sum_{\ell=1}^{m} ( \bSigma^*_{\ell} - \bV_K \bV_K^T ) \right \|_F
\notag\\ &
\leq  \left\| \frac{1}{m} \sum_{\ell=1}^{m} ( \bSigma^*_{\ell} - \bar\bV_K^{(\ell)} \bar\bV_K^{(\ell)T} ) \right \|_F
+ \left\| \frac{1}{m} \sum_{\ell=1}^{m} ( \bar\bV_K^{(\ell)} \bar\bV_K^{(\ell)T} - \bV_K \bV_K^T ) \right \|_F
.\label{ineq-factor-2}
\end{align}
Similar to \eqref{ineq-common-principal-1.5} and \eqref{ineq-common-principal-2}, with $n\geq r_{\ell}$ we have $\| \bSigma^*_{\ell} - \bar\bV_K^{(\ell)} \bar\bV_K^{(\ell)T} \|_F \lesssim B_{\ell}$ and
\begin{align}
\left\| \frac{1}{m} \sum_{\ell=1}^{m} ( \bSigma^*_{\ell} - \bar\bV_K^{(\ell)} \bar\bV_K^{(\ell)T} ) \right \|_F
\leq  \frac{1}{m} \sum_{\ell=1}^{m} \left\|  \bSigma^*_{\ell} - \bar\bV_K^{(\ell)} \bar\bV_K^{(\ell)T}  \right \|_F \lesssim  \frac{1}{m} \sum_{\ell=1}^{m} B_{\ell}
.\label{ineq-factor-3}
\end{align}
For the last part in \eqref{ineq-factor-2}, note that
$\bar\bV_K^{(\ell)}$ and $\bV_K$ contain eigenvectors of $\bSigma^{(\ell)}$ and $\bB^{(\ell)} \bB^{(\ell)T}$. Hence the Davis-Kahan theorem forces
\begin{align*}
\| \bar\bV_K^{(\ell)} \bar\bV_K^{(\ell)T} - \bV_K \bV_K^T \|_F
\lesssim \frac{ \sqrt{K} \| \bSigma_u^{(\ell)} \|_2  }{ \lambda_K( \bLambda_K^{(\ell)} ) }
.
\end{align*}
and
\begin{align}
\left\| \frac{1}{m} \sum_{\ell=1}^{m} ( \bar\bV_K^{(\ell)} \bar\bV_K^{(\ell)T} - \bV_K \bV_K^T ) \right \|_F
\lesssim \frac{\sqrt{K}}{m} \sum_{\ell=1}^{m}\frac{  \| \bSigma_u^{(\ell)} \|_2  }{ \lambda_K( \bLambda_K^{(\ell)} ) }
.\label{ineq-factor-4}
\end{align}
The proof is completed by collecting \eqref{ineq-factor-1}, \eqref{ineq-factor-2}, \eqref{ineq-factor-3} and \eqref{ineq-factor-4}.
\end{proof}

\subsection{Technical lemmas}
\subsubsection{Tail bounds}
\begin{lem}\label{lem-cov-concentration}
	Suppose $\bX $ and $\{\bX_i\}_{i=1}^n$ are i.i.d. sub-Gaussian random vectors in $\R^d$ with zero mean and covariance matrix $\bSigma\succeq 0$. Let $\hat{\bSigma} = \frac{1}{n} \sum_{i=1}^{n} \bX_i \bX_i^T$ be the sample covariance matrix, $\{\lambda_j \}_{j=1}^d$ be the eigenvalues of $\bSigma$ sorted in descending order, and $r=\Tr(\bSigma)/\|\bSigma\|_2$. There exist constants $c\geq 1$ and $C\geq 0$ such that when $n \geq r$, we have
	\[
	\P\left(
	\|\hat{\bSigma}-\bSigma\|_{2} \geq s \right)\leq
	\exp\left( - \frac{ s }{c \lambda_1 \sqrt{r/n} } \right),\qquad \forall s \geq 0,
	\]
	and
	$\left\| \|\hat{\bSigma}-\bSigma\|_{2} \right\|_{\psi_1} \leq C \lambda_1 \sqrt{r/n}$.
\end{lem}
\begin{proof}
	By the Theorem 9 in \cite{KLo17} and the simple fact
	\[
	(\E \ltwonorm{\bX})^2/\|\bSigma\|_{2} \leq \E \ltwonorm{\bX}^2 / \|\bSigma\|_{2} = \Tr(\bSigma)/\|\bSigma\|_2 = r(\bSigma),
	\]
	we know the existence of a constant $c\geq 1$ such that
	\begin{equation}\label{ineq-cov-concentration}
	\P\left(
	\|\hat{\bSigma}-\bSigma\|_{2} \geq c\lambda_1\max\left\{\sqrt{\frac{r}{n}},\frac{r}{n},\sqrt{\frac{t}{n}},\frac{t}{n}\right\} \right) \leq e^{-t},\qquad \forall t\geq 1.
	\end{equation}
	
Since $1 \leq r \leq n$, (\ref{ineq-cov-concentration}) yields
	\begin{align}
	&\P\left(
	\|\hat{\bSigma}-\bSigma\|_{2} \geq c\lambda_1 \sqrt{\frac{t}{n}} \right) \leq e^{-t},\qquad r \leq t \leq n,\label{ineq-cov-concentration-1}\\
	&\P\left(
	\|\hat{\bSigma}-\bSigma\|_{2} \geq c\lambda_1 \frac{t}{n} \right) \leq e^{-t},\qquad  t \geq n.
	\label{ineq-cov-concentration-2}
	\end{align}
	When $r \leq t \leq n$, we have $\sqrt{\frac{t}{n}} \leq \frac{t}{n}\sqrt{\frac{n}{r}}$. By letting $s=c\lambda_{1} \frac{t}{n} \sqrt{\frac{n}{r}}$ we derive from (\ref{ineq-cov-concentration-1}) that for $c \lambda_{1} \sqrt{\frac{r}{n}} \leq s \leq c\lambda_{1} \sqrt{\frac{n}{r}}$,
	\begin{align}
	&\P\left(
	\|\hat{\bSigma}-\bSigma\|_{2} \geq s \right)
	\leq
	\P\left(
	\|\hat{\bSigma}-\bSigma\|_{2} \geq  c\lambda_{1} \sqrt{\frac{t}{n}} \right) \leq e^{-t}
	=\exp\left( - \frac{ s \sqrt{nr} }{c \lambda_{1} } \right)
	.\label{ineq-cov-concentration-3}
	\end{align}
	When $t\geq n$, we let $s=c\lambda_{1} \frac{t}{n}$ and derive from (\ref{ineq-cov-concentration-2}) that for $s \geq c \lambda_{1}$,
	\begin{align}
	&\P\left(
	\|\hat{\bSigma}-\bSigma\|_{2} \geq s \right)
	=
	\P\left(
	\|\hat{\bSigma}-\bSigma\|_{2} \geq  c\lambda_{1} \frac{t}{n} \right) \leq e^{-t}
	=\exp\left( - \frac{ ns }{c \lambda_{1} } \right)
	.\label{ineq-cov-concentration-4}
	\end{align}
	(\ref{ineq-cov-concentration-3}), (\ref{ineq-cov-concentration-4}) and $n\geq r$ lead to
	\[
	\P\left(
	\|\hat{\bSigma}-\bSigma\|_{2} \geq s \right)\leq
	\exp\left( - \frac{ s \sqrt{nr} }{c \lambda_1 } \right),\qquad \forall s\geq c \lambda_1 \sqrt{r/n}.
	\]
	and thus
	\[
	\P\left(
	\|\hat{\bSigma}-\bSigma\|_{2} \geq s \right)\leq
	\exp\left( 1- \frac{ s }{c \lambda_1 \sqrt{r/n} } \right),\qquad \forall s\geq 0.
	\]
	According to the Definition 5.13 in \cite{Ver10}, we get $\left\| \|\hat{\bSigma}-\bSigma\|_{2} \right\|_{\psi_1} \leq C \lambda_1 \sqrt{r/n}$ for some constant $C$.
\end{proof}
The next lemma investigates the sum of independent random vectors in a Hilbert space whose norms are sub-exponential, which directly follows from Theorem 2.5 in \cite{Bos00}.
\begin{lem}\label{lem-Bosq}
If $\{ \bX_i \}_{i=1}^n$ are independent random vectors in a separable Hilbert space (where the norm is denoted by $\|\cdot\|$) with $\E \bX_i = \bzero$ and $\left\| \| \bX_i \| \right\|_{\psi_1} \leq L_i<\infty$. We have
\[
\left\| \bigg \| 
\sum_{i=1}^{n} \bX_i
\bigg\| \right\|_{\psi_1} \lesssim
\sqrt{\sum_{i=1}^{n} L_i^2 }
.
\]
\end{lem}
\begin{proof}
We are going to apply Theorem 2.5 in \cite{Bos00}. By definition $k^{-1} \E^{1/k} \| \bX_i \|^k \leq  \left\| \| \bX_i \| \right\|_{\psi_1} \leq L_i$ for all $k\geq 1$, and
\[
\E \| \bX_i \|^k \leq \left(k L_i \right)^k \leq \sqrt{2\pi k} (k/e)^k \left(e L_i \right)^k \lesssim k!\left(e L_i \right)^k.
\]
Hence there exists some constant $c$ such that $\E \| \bX_i \|^k \leq \frac{k!}{2} \left(c L_i \right)^k$ for $k\geq 2$. Let $\ell = \sqrt{ c^2 \sum_{i=1}^{n} L_i^2 }$ and $b=c\cdot \max_{i\in[n]} L_i$. We have
\begin{align*}
\sum_{i=1}^{n} \E \| \bX\|^k & \leq \frac{k!}{2} \sum_{i=1}^{n} \left( c L_i \right)^k
\leq \frac{k!}{2} \left( \sum_{i=1}^{n} c^2 L_i^2
\right) \left( c \cdot \max_{i\in[n]} L_i \right)^{k-2}
= \frac{k!}{2} \ell^2 b^{k-2},\qquad \forall k\geq 2.
\end{align*}
Let $\bS_n = \sum_{i=1}^{n} \bX_i$. Theorem 2.5 in \cite{Bos00} implies that
\[
\P \left(  \| \bS_n \| \geq t \right) \leq 2 \exp \left( - \frac{t^2}{2\ell^2 + 2 bt} \right)
,\qquad \forall t >0.
\]
When $4 \ell \leq t \leq \ell^2/b$ (this cannot happen if $4b > \ell$), we have $2\ell^2 \geq 2bt$ and
\begin{align*}
\P \left(  \| \bS_n \| \geq t \right)
&\leq 2 \exp \left( - \frac{t^2}{2\ell^2+2\ell^2} \right)
\leq  2 \exp \left( - \frac{4\ell\cdot t}{4\ell^2} \right)
= 2 \exp \left( -  \frac{t}{\ell} \right)
\leq \exp\left( 1 - \frac{t}{4\ell  } \right).
\end{align*}
When $ t\geq \ell^2/b$, we have $2bt \geq 2\ell^2$ and
\begin{align*}
\P \left(  \| \bS_n \| \geq t \right)
&\leq 2 \exp \left( - \frac{t^2}{2bt+2bt} \right)
= 2 \exp \left( - \frac{t}{4b} \right)
\leq \exp\left( 1 - \frac{t}{4\ell } \right),
\end{align*}
where the last inequality follows from $2\leq e$ and $b\leq \ell$.
It is then easily seen that
\begin{align*}
\P \left(  \| \bS_n \| \geq t \right) \leq \exp\left( 1 - \frac{t}{4 \ell } \right),\qquad\forall t\geq 0.
\end{align*}
With the help of Definition 5.13 in \cite{Ver10}, we can conclude that
\[
\left\| \| \bS_n \| \right\|_{\psi_1} \lesssim \ell
\lesssim
\sqrt{ \sum_{i=1}^{n} L_i^2 }.
\]
\end{proof}

\subsubsection{Matrix analysis}
\begin{lem}\label{lem-SDP}
Suppose that $\bA\in\R^{d\times d}$ is a symmetric matrix with eigenvalues $\{ \lambda_j\}_{j=1}^d$ (in descending order) and corresponding eigenvectors $\{ \bu_j \}_{j=1}^d$. When $K\in[d]$, $\bP_K=\sum_{j=1}^{K} \bu_j \bu_j^T$ is an optimal solution to the SDP:
\begin{equation}\label{SDP-lem}
\begin{split}
&\min_{\bP \in S^{d\times d}} -\Tr( \bP^T \bA )\\
&\text{s.t. } \Tr(\bP)\leq K, \| \bP \|_2 \leq 1, \bP \succeq 0.
\end{split}
\end{equation}
\end{lem}
\begin{proof}
By orthonormal invariance of the problem formulation, we assume without loss of generality that $\{ \bu_j \}_{j=1}^d$ are the canonical bases $\{ \be_j \}_{j=1}^d$. Then $\bA=\diag(\lambda_1,\cdots\lambda_d)$ and $\Tr( \bP^T \bA ) = \sum_{j=1}^{d} \lambda_j P_{jj}$. The constraints on $\bP$ force $0\leq P_{jj} \leq 1$ and $\sum_{j=1}^{d} P_{jj}\leq K$. Hence $-\Tr( \bP^T \bA ) \geq -\sum_{j=1}^{K} \lambda_j$ always holds, and $\bP_K=\sum_{j=1}^{K} \be_j \be_j^T$ is a feasible solution that attains this minimum.
\end{proof}

\begin{lem}\label{lem-loss}
Suppose $\widehat{\bV}_K^{(\ell)}\in\cO_{d\times K}$, $\forall \ell\in[m]$, and define $\widetilde{\bSigma}=\frac{1}{m}\sum_{\ell=1}^{m}\widehat{\bV}_K^{(\ell)} \widehat{\bV}_K^{(\ell)T}$. Let $\widetilde{\bSigma}=\sum_{j=1}^{d}\tilde{\lambda}_j\tilde{\bv}_j\tilde{\bv}_j^T$ be its eigen-deconposition, where $\tilde{\lambda}_1\geq\cdots\geq\tilde{\lambda}_d$. Then $\widetilde{\bV}_K=(\tilde{\bv}_1,\cdots,\tilde{\bv}_K)\in\argmin_{\bU\in\cO_{d\times K}}\sum_{\ell=1}^{m}\rho^2(\bU,\widehat{\bV}_K^{(\ell)})$.
\end{lem}
\begin{proof}
Let $\widehat{\bP}^{(\ell)}=\widehat{\bV}_K^{(\ell)} \widehat{\bV}_K^{(\ell)T}$ and $R(\bU)=\frac{1}{m}\sum_{\ell=1}^{m}\rho^2(\bU,\widehat{\bV}_K^{(\ell)})$. Then $\widetilde{\bSigma}=\frac{1}{m}\sum_{\ell=1}^{m}\widehat{\bP}^{(\ell)}$ and
	\begin{equation*}
	\begin{split}
	& R(\bU)=\frac{1}{m}\sum_{\ell=1}^{m}\rho^2(\bU,\widehat{\bV}_K^{(\ell)})=\frac{1}{m}\sum_{\ell=1}^{m}\|\bU\bU^T-\hat{\bP}^{(\ell)}\|_F^2\\
	&=\|\bU\bU^T-\widetilde{\bSigma}\|_F^2+\frac{1}{m}\sum_{\ell=1}^{m}\|\widetilde{\bSigma}-\hat{\bP}^{(\ell)}\|_F^2\\
	&=\|\bU\bU^T\|_F^2+\|\widetilde{\bSigma}\|_F^2-2\Tr(\bU \bU^T\widetilde{\bSigma})+\frac{1}{m}\sum_{\ell=1}^{m}\|\widetilde{\bSigma}-\hat{\bP}^{(\ell)}\|_F^2.
	\end{split}
	\end{equation*}
	The fact $\bU\in\cO_{d\times K}$ forces $\|\bU\bU^T\|_F^2=K$. Hence
	\begin{equation*}
	\argmin_{\bU\in\cO_{d\times K}}R(\bU)=\argmax_{\bU\in\cO_{d\times K}}\Tr(\bU \bU^T\widetilde{\bSigma}).
	\end{equation*}
By slightly modifying the proof for Lemma \ref{lem-SDP} we get the desired result.
\end{proof}

Suppose that $\bU,\bV\in\O_{d\times K}$. Let $\bP_{\bU} = \bU \bU^T$, $\bP_{\bV} = \bV \bV^T$, $\bH=\bV^T\bU$, and $\{ \sigma_j \}_{j=1}^K$ be the singular values (sorted in descending order) of $\bH$. By the Corollary 5.4 in Chapter I, \cite{Ste90}, $\{ \sigma_j \}_{j=1}^K$ are cosines of the canonical angles $\{ \theta_j \}_{j=1}^K \subseteq [0,\pi/2)$ between $\mathrm{Col}(\bU)$ and $\mathrm{Col}(
\bV)$. Let $\sin\bTheta( \bU , \bV)=\diag(\sin\theta_1, \cdots, \sin\theta_K)$.

Define $\hat{\bH}=\sgn(\bH)$. Here $\sgn(\cdot)$ is the matrix sign function (see \cite{Gro11}) defined as follows: let $\bH=\sum_{j=1}^{K} \sigma_{j} \tilde{\bu}_j \tilde{\bv}_j^T$ be the singular value decomposition, where $\{ \tilde{\bu}_j \}_{j=1}^K$, $\{ \tilde{\bv}_j \}_{j=1}^K$ are two orthonormal bases in $\R^K$ and $\{ \sigma_j \}_{j=1}^K \subseteq [0,+\infty)$, then $\hat{\bH} = \sum_{j=1}^{K} \sgn(\sigma_j)  \tilde{\bu}_j \tilde{\bv}_j^T = \sum_{\sigma_j >0} \tilde{\bu}_j \tilde{\bv}_j^T$.

\begin{lem}\label{2ndDK-lem-projection}
	We have
	$\| \bP_{\bU} - \bP_{\bV} \|_2 = \| \sin \bTheta ( \bU,\bV ) \|_2$ and $ \| \bP_{\bU} - \bP_{\bV} \|_F=\sqrt{2} \| \sin \bTheta ( \bU,\bV ) \|_F $.
	If $\| \bP_{\bU} - \bP_{\bV} \|_2 <1$, then $\hat{\bH}$ is orthonormal, $\| \bH - \hat{\bH } \|_2
	\leq \frac{\| \bP_{\bU} - \bP_{\bV} \|_2^2}{2-\| \bP_{\bU} - \bP_{\bV} \|_2^2}$,
	\begin{align*}
	& \| \bV \bH - \bU\|_F\leq \| \bV\hat{\bH} - \bU\|_F = \sqrt{2} \| \bH - \hat{\bH} \|_{*}^{1/2},\\
	&\frac{1}{\sqrt{2}} \| \bP_{\bU} - \bP_{\bV} \|_F \leq \| \bV\hat{\bH} - \bU\|_F
	\leq \frac{ \| \bP_{\bU} - \bP_{\bV} \|_F }{\sqrt{2 - \| \sin \bTheta ( \bU,\bV ) \|_2^2 }} .
	\end{align*}
\end{lem}

\begin{proof}
	By the Theorem 5.5 in Chapter I, \cite{Ste90}, the singular values of $ \bP_{\bU} - \bP_{\bV}$ are $\sin \theta_K, \sin \theta_K, \sin \theta_{K-1}, \sin \theta_{K-1}, \cdots, \sin \theta_1, \sin \theta_1, 0,\cdots,0$. This immediately leads to $\| \bP_{\bU} - \bP_{\bV} \|_2 = \| \sin \bTheta ( \bU,\bV ) \|_2$ and $\| \bP_{\bU} - \bP_{\bV} \|_F = \sqrt{2}\| \sin \bTheta ( \bU,\bV ) \|_F$.
	
	When $\| \sin \bTheta ( \bU,\bV ) \|_2 = \| \bP_{\bU} - \bP_{\bV}\|_2<1$, we have $\theta_K<\pi/2$. Thus the smallest singular value of $\bH$ is $\sigma_K = 1-\cos \theta_K >0$, and $\hat{\bH}$ is orthonormal. Observe that
	\begin{align}\label{2ndDK-eqn-1213}
	\| \bV \hat{\bH} -\bU\|_F^2 & =\|\bV \hat{\bH}\|_F^2+\|\bU\|_F^2-2\Tr(\hat{\bH}^T \bV^T \bU)=2K-2\Tr(\hat{\bH}^T\bH) \notag \\
	&= 2 \sum_{j=1}^K ( 1 - \sigma_j) = 2 \| \bH- \hat{\bH} \|_{*}.
	\end{align}
	Hence $\| \bV \bH -\bU\|_F \leq \| \bV \hat{\bH} -\bU\|_F$ follows from
	\begin{align*}
	\| \bV \bH -\bU\|_F^2 &=\|\bV \bH\|_F^2+\|\bU\|_F^2-2\Tr(\bH^T \bV^T \bU)=K-\| \bH\|_F^2 \\
	&= \sum_{j=1}^K ( 1 - \sigma_j^2) =
	\sum_{j=1}^K ( 1 - \sigma_j)(1+\sigma_j) \leq 2 \sum_{j=1}^K ( 1 - \sigma_j)=\| \bV \hat{\bH} -\bU\|_F^2.
	\end{align*}
	
	For any $\theta\in[0,\pi/2)$, we have $1- \cos \theta = \frac{1 - \cos^2 \theta}{1+\cos \theta} = \frac{\sin^2 \theta}{2-(1-\cos\theta)}$, which leads to $\frac{1}{2} \sin^2 \theta \leq 1-\cos\theta \leq \sin^2 \theta$ and furthermore, $1 - \cos\theta = \frac{\sin^2 \theta}{2-(1-\cos\theta)} \leq  \frac{\sin^2 \theta}{2-\sin^2 \theta}$. Hence
	\begin{align*}
	&\| \bV \hat{\bH} -\bU\|_F^2 = 2 \sum_{j=1}^K ( 1 - \sigma_j) = 2 \sum_{j=1}^K ( 1 - \cos\theta_j) \geq \sum_{j=1}^K \sin^2\theta_j = \frac{1}{2} \| \bP_{\bU} - \bP_{\bV} \|_F^2,\\
	&\| \bV \hat{\bH} -\bU\|_F^2 \leq  2 \sum_{j=1}^K \frac{ \sin^2\theta_j }{2-\sin^2\theta_j}
	\leq \frac{2 \sum_{j=1}^K \sin^2\theta_j }{2-\sin^2\theta_K}
	= \frac{ \| \bP_{\bU} - \bP_{\bV} \|_F^2 }{ 2-\| \sin \bTheta ( \bU,\bV ) \|_2^2 } ,\\
	&\| \bH - \hat{\bH} \|_2 = 1-\sigma_K =1- \cos \theta_K \leq \frac{\sin^2 \theta_K}{2-\sin^2\theta_K}
	=\frac{\| \sin \bTheta ( \bU,\bV ) \|_2^2}{2-\| \sin \bTheta ( \bU,\bV ) \|_2^2}.
	\end{align*}
\end{proof}

\begin{lem}
	\label{lem:DK-strong}
Consider the settings for Lemma \ref{lem:2} and define $\bH = \hat{\bU}^T \bU$ and $\hat{\bH} = \sgn(\bH)$.	When $\varepsilon=\|\bE\|_{2}/\Delta \leq 1/10$, we have $\| \hat{\bU} \hat{\bU}^T - \bU \bU^T \|_2 \leq \varepsilon / (1-\varepsilon)$,
	\begin{align}
	&\frac{\| f(\bE \bU) \|_F}{1 + 5 \varepsilon} \leq
	\| \hat{\bU}\hat{\bH}-\bU \|_F \leq  \frac{ \| f(\bE \bU) \|_F }{1- 5 \varepsilon},
	\label{lem:DK-strong-1}\\
	&\|\hat{\bU}\hat{\bH}-\bU-f(\bE\bU)\|_F\leq 9 \varepsilon \| f(\bE \bU) \|_F,
	\label{lem:DK-strong-2}\\
	&\frac{ \sqrt{2} \| f(\bE \bU ) \|_F }{1 + 7\varepsilon} \leq \| \hat{\bU} \hat\bU^T -  \bU \bU^T \|_F \leq  \frac{\sqrt{2} \| f(\bE \bU ) \|_F}{1-7\varepsilon},
	\label{lem:DK-strong-3}\\
	&\|  \hat{\bU} \hat{\bU}^T  - \bU \bU^T - [ f(\bE\bU)\bU^T+\bU f(\bE \bU)^T]\|_F
	\leq 24 \varepsilon \| f(\bE\bU)\|_F.\label{lem:DK-strong-4}
	\end{align}
	Besides, $\| f(\bE\bU)\bU^T+\bU f(\bE \bU)^T \|_F =\sqrt{2} \| f(\bE \bU ) \|_F$.
\end{lem}

\begin{proof}
	Define $\hat{\bP} = \hat{\bU} \hat{\bU}^T$, $\bP = \bU\bU^T$ and $\bP_{\perp} = \bI- \bP$. The Davis-Kahan $\sin\bTheta$ theorem \citep{DKa70} and Lemma \ref{2ndDK-lem-projection} force that $\delta  \| \hat{\bP} - \bP \|_2 \leq \| \bE \bP \|_2 \leq \| \bE \|_2$, where $\delta = \min\{  ( \hat{\lambda}_s - \lambda_{s+1} )_{+}, ( \lambda_{s+r} - \hat{\lambda}_{s+r+1} )_{+} \}$ and we define $x_+ = \max\{ x,0 \}$ for $x\in\R$.
	Since the Weyl's inequality \cite[Corollary IV.4.9]{Ste90} leads to $\delta \geq \Delta - \| \bE\|_2 = (1-\varepsilon) \Delta$, we get $ \| \hat{\bP} - \bP \|_2 \leq \varepsilon / (1-\varepsilon)$.
	
	To attack (\ref{lem:DK-strong-1}) and (\ref{lem:DK-strong-2}), we divide the difference
	\begin{equation}\label{2ndDK-bound-1}
	\hat{\bU}\hat{\bH}-\bU-f(\bE \bU)=[\bP_{\perp}\hat{\bU}\bH - f(\bE \bU) ]+\bP_{\perp}\hat{\bU}(\hat{\bH}-\bH)+(\bP\hat{\bU}\hat{\bH}-\bU)
	\end{equation}
	and conquer the terms separately. Since $\varepsilon<1/2$, the first claim in Lemma \ref{lem:DK-strong} yields $\| \hat{\bU} \hat{\bU}^T - \bU \bU^T \|_2 <1$. Then according to Lemma \ref{2ndDK-lem-projection}, $\hat{\bH}$ is orthonormal,
	\begin{align*}
	&\|\bP_{\perp}\hat{\bU}(\hat{\bH}-\bH)\|_F
	=\| \bP_{\perp} ( \hat{\bU} \hat{\bH} - \bU ) \hat{\bH}^T (\hat{\bH}-\bH) \|_F
	\leq \|  \hat{\bU} \hat{\bH} - \bU \|_F \|\hat{\bH}-\bH\|_2,
	\end{align*}
	and
	\begin{align}
	\|\bP\hat{\bU}\hat{\bH}-\bU\|_F
	&=\|\bU( \bH^T \hat{\bH} - \bI ) \|_F \leq \| \bH^T \hat{\bH} - \bI  \|_F
	=\| (\bH - \hat{\bH})^T \hat{\bH} \|_F
	=\| \bH - \hat{\bH} \|_F \notag \\
	&\leq \| \bH - \hat{\bH} \|_2^{1/2} \| \bH - \hat{\bH} \|_{*}^{1/2}=
	\| \bH - \hat{\bH} \|_2^{1/2} \| \hat{\bU} \hat{\bH} -\bU \|_F /\sqrt{2}.
	\label{2ndDK-bound-1.5}
	\end{align}
	Observe that when $\varepsilon \leq 1/10$, Lemma \ref{2ndDK-lem-projection} forces that
	\begin{align}
	\| \bH - \hat{\bH} \|_2
	\leq \frac{\| \hat{\bP} - \bP  \|_2^2}{2-\| \hat{\bP} - \bP  \|_2^2}
	\leq \frac{\left( \frac{\varepsilon}{1-\varepsilon} \right)^2 }{ 2 -\left( \frac{\varepsilon}{1-\varepsilon} \right)^2 }
	= \frac{ \varepsilon^2 }{ 2 (1-\varepsilon)^2 - \varepsilon^2 }
	\leq  \frac{5}{8} \varepsilon^2 \leq \frac{1}{16} \varepsilon.
	\label{2ndDK-bound-H}
	\end{align}
	Combining the estimates above yields
	\begin{align}\label{2ndDK-bound-2}
	\|\bP_{\perp}\hat{\bU}(\hat{\bH}-\bH) +  ( \bP \hat{\bU}\hat{\bH}-\bU )\|_F
	\leq \left( \frac{1}{16} +\sqrt{\frac{5}{16}} \right) \varepsilon \|  \hat{\bU} \hat{\bH} - \bU \|_F
	\leq \frac{16}{25} \varepsilon \|  \hat{\bU} \hat{\bH} - \bU \|_F.
	\end{align}

	We start to work on $\bP_{\perp}\hat{\bU}\bH-f(\bE \bU)$. Define $\bLambda=\diag(\lambda_{s+1},\cdots,\lambda_{s+K})$, and $L(\bV)=\bA \bV-\bV\bLambda$ for $\bV\in\R^{d\times K}$. Note that $L(\bv_1,\cdots,\bv_K)=((\bA-\lambda_{s+1}\bI)\bv_1,\cdots,(\bA-\lambda_{s+K}\bI)\bv_K)$, and $\bG_j(\bA-\lambda_{s+j}\bI)=\bP_{\perp}$ holds for all $j\in [K]$. As a result, $f(L(\bV))=-\bP_{\perp}\bV$ for any $\bV\in\R^{d\times K}$. This motivates us to work on $L(\hat{\bU} \bH)$ in order to study $\bP_{\perp} \hat{\bU} \bH$.

	Let $\hat{\bLambda}=\diag(\hat{\lambda}_{s+1},\cdots,\hat{\lambda}_{s+K})$. By definition, $\hat{\bA}\hat{\bU}=\hat{\bU}\hat{\bLambda}$ and
	\begin{align}
	L(\hat{\bU}\bH) & =\bA \hat{\bU}\bH-\hat{\bU}\bH\bLambda\notag \\
	&=(\bA-\hat{\bA}) \hat{\bU}\bH +(\hat{\bA}\hat{\bU}-\hat{\bU}\hat{\bLambda})\bH + \hat{\bU}(\hat{\bLambda}-\bLambda)\bH+ \hat{\bU}(\bLambda\bH-\bH \bLambda) \notag \\
	&=-\bE \hat{\bU}\bH  + \hat{\bU}(\hat{\bLambda}-\bLambda)\bH + \hat{\bU}(\bLambda\bH-\bH\bLambda).
	\label{2ndDK-decomposition-1212}
	\end{align}
	
	Now we study the images of these three terms under the linear mapping $f$. First, the facts $\|f(\cdot)\|_F\leq \Delta^{-1}\|\cdot\|_F$ and $\|\hat{\bU} \bH-\bU\|_F \leq \|\hat{\bU} \hat{\bH}-\bU\|_F$ (by Lemma \ref{2ndDK-lem-projection}) imply that
	\begin{align}\label{2ndDK-term1}
	& \|f(\bE\hat{\bU}\bH)-f(\bE\bU)\|_F = \|f [ \bE ( \hat{\bU}\bH - \bU) ]\|_F \leq \Delta^{-1}\|\bE(\hat{\bU}\bH-\bU)\|_F \notag\\
	&\leq
	\Delta^{-1}\|\bE\|_2 \|\hat{\bU}\bH-\bU\|_F
	\leq \varepsilon \|\hat{\bU} \hat{\bH}-\bU\|_F.
	\end{align}
	Second, the definition of $f$ forces $f(\bU \bM)=\mathbf{0}$ for all $\bM\in\R^{K\times K}$.
	\begin{align}
	&\|f[\hat{\bU}(\hat{\bLambda}-\bLambda)\bH]\|_F=\|f[(\hat{\bU}\hat{\bH}-\bU)\hat{\bH}^{T} (\hat{\bLambda}-\bLambda)\bH]\|_F \notag \\
	& \leq \Delta^{-1} \|(\hat{\bU} \hat{\bH}-\bU)\hat{\bH}^{T} (\hat{\bLambda}-\bLambda)\bH\|_F
	\leq \Delta^{-1} \|\hat{\bU}\hat{\bH}-\bU\|_F\|  \hat{\bH}^T (\hat{\bLambda}-\bLambda)\bH\|_{2}\notag \\
	&\leq \Delta^{-1} \|\hat{\bU}\hat{\bH}-\bU\|_F  \|\bE \|_{2}\|\bH\|_{2}
	\leq \varepsilon \|\hat{\bU} \hat{\bH}-\bU\|_F.\label{2ndDK-term2}
	\end{align}
	Here we applied Weyl's inequality $\|\hat{\bLambda}-\bLambda\|_{2}\leq\|\bE\|_{2}$ and used the fact that $\|\bH\|_{2}=\|\hat{\bU}^T\bU\|_{2}\leq 1$. Third, by similar tricks we work on the third term
	\begin{align}
	&\|f[\hat{\bU} (\bLambda\bH-\bH\bLambda)]\|_F=\|f[(\hat{\bU} \hat{\bH}-\bU) \hat{\bH}^T (\bLambda\bH-\bH \bLambda)]\|_F\notag \\
	&\leq \Delta^{-1} \|(\hat{\bU} \hat{\bH} -\bU) \hat{\bH}^T (\bLambda\bH-\bH \bLambda)\|_F
	\leq \Delta^{-1} \|\hat{\bU} \hat{\bH} -\bU\|_F \|\bLambda\bH-\bH \bLambda\|_{2}.
	\label{2ndDK-term3}
	\end{align}
	As an intermediate step, we are going to show that $\|\bLambda\bH-\bH \bLambda\|_{2}\leq 2\|\bE\|_{2}$. On the one hand, $\hat{\bA}\hat{\bU}=\hat{\bU}\hat{\bLambda}$ yields
	\begin{equation}\label{2ndDK-intermediate-L-1212}
	\begin{split}
	&L(\hat{\bU})
	=(\bA-\hat{\bA}) \hat{\bU} +(\hat{\bA}\hat{\bU}-\hat{\bU}\hat{\bLambda}) + \hat{\bU}(\hat{\bLambda}-\bLambda)
	=-\bE \hat{\bU}  + \hat{\bU}(\hat{\bLambda}-\bLambda).\\
	\end{split}
	\end{equation}
	On the other hand, let $\bU_1=(\bu_1,\cdots,\bu_s,\bu_{s+K+1},\cdots,\bu_d)$, $ \hat{\bU}_1=(\hat{\bu}_1,\cdots,\hat{\bu}_s,\hat{\bu}_{s+K+1},\cdots,\hat{\bu}_d)$, and $\bLambda_1=\diag(\lambda_1,\cdots,\lambda_s,\lambda_{s+K+1},\cdots,\lambda_d)$.
	We have
	\begin{equation*}
	\begin{split}
	&\bA \hat{\bU}=
	\begin{pmatrix}
	\bU & \bU_1
	\end{pmatrix}
	\begin{pmatrix}
	\bLambda & \mathbf{0}\\
	\mathbf{0} & \bLambda_1
	\end{pmatrix}
	\begin{pmatrix}
	\bU^T\\
	\bU_1^T
	\end{pmatrix}
	\hat{\bU}
	=\begin{pmatrix}
	\bU & \bU_1
	\end{pmatrix}
	\begin{pmatrix}
	\bLambda \bH^T\\
	\bLambda_1 \bU_1^T \hat{\bU}
	\end{pmatrix}
	,\\
	&\hat{\bU} \bLambda= \begin{pmatrix}
	\bU & \bU_1
	\end{pmatrix}
	\begin{pmatrix}
	\bU^T \\
	\bU_1^T
	\end{pmatrix}
	\hat{\bU} \bLambda=
	\begin{pmatrix}
	\bU & \bU_1
	\end{pmatrix} \begin{pmatrix}
	\bH^T \bLambda\\
	\bU_1^T \hat{\bU} \bLambda
	\end{pmatrix}
	.
	\end{split}
	\end{equation*}
	As a result, (\ref{2ndDK-intermediate-L-1212}) yields that
	\begin{equation}\label{2ndDK-term0}
	\begin{split}
	&\|\bLambda \bH-\bH \bLambda\|_{2}
	=\|\bH^T \bLambda -\bLambda \bH^T \|_{2}
	\leq \|L(\hat{\bU})\|_{2}= \|-\bE \hat{\bU}+\hat{\bU}(\hat{\bLambda}-\bLambda)\|_{2}\leq 2\|\bE\|_{2}.
	\end{split}
	\end{equation}
	By combining (\ref{2ndDK-decomposition-1212}), (\ref{2ndDK-term1}), (\ref{2ndDK-term2}), (\ref{2ndDK-term3}) and (\ref{2ndDK-term0}), we obtain that
	\begin{equation}\label{2ndDK-bound-3}
	\begin{split}
	&\|\bP_{\perp}\hat{\bU}\bH-f(\bE\bU)\|_F=\|-f[L(\hat{\bU}\bH)]-f(\bE\bU)\|_F
	\leq 4 \varepsilon \|\hat{\bU} \hat{\bH}-\bU\|_F.
	\end{split}
	\end{equation}
	Based on (\ref{2ndDK-bound-1}), (\ref{2ndDK-bound-2}) and (\ref{2ndDK-bound-3}), we obtain that
	\begin{equation}\label{2ndDK-bound-4}
	\| \hat{\bU}\hat{\bH}-\bU-f(\bE \bU) \|_F \leq \frac{116}{25}\varepsilon \| \hat{\bU} \hat{\bH} - \bU \|_F.
	\end{equation}
	It follows from the triangle's inequality that
	\begin{align*}
	&\frac{\| f(\bE \bU) \|_F}{1 + 5 \varepsilon } \leq
	\frac{\| f(\bE \bU) \|_F}{1 + 116 \varepsilon /25 } \leq
	\| \hat{\bU}\hat{\bH}-\bU \|_F \leq  \frac{ \| f(\bE \bU) \|_F }{1- 116 \varepsilon /25}
	\leq  \frac{ \| f(\bE \bU) \|_F }{1- 5 \varepsilon},\\
	&\| \hat{\bU}\hat{\bH}-\bU-f(\bE \bU) \|_F \leq \frac{116}{25}\varepsilon
	\frac{ \| f(\bE \bU) \|_F }{1-\frac{116}{25}\cdot\frac{1}{10}}
	\leq 8.66 \varepsilon \| f(\bE \bU) \|_F
	\leq 9 \varepsilon \| f(\bE \bU) \|_F.
	\end{align*}
	
	Hence we have proved (\ref{lem:DK-strong-1}) and (\ref{lem:DK-strong-2}). Now we move on to (\ref{lem:DK-strong-3}) and (\ref{lem:DK-strong-4}). Note that
	\begin{align}
	\hat{\bP} - \bP &= \hat{\bU} \hat{\bH} (\hat{\bU} \hat{\bH})^T - \bU \bU^T
	= (\hat{\bU} \hat{\bH} - \bU) (\hat{\bU} \hat{\bH})^T + \bU (\hat{\bU} \hat{\bH} - \bU)^T\notag\\
	&=(\hat{\bU} \hat{\bH} - \bU) ( \hat{\bU} \hat{\bH} - \bU )^T + (\hat{\bU} \hat{\bH} - \bU) \bU^T + \bU (\hat{\bU} \hat{\bH} - \bU)^T.\label{2ndDK-bound-5}
	\end{align}
	The first term is controlled by
	\begin{align}
	&\| (\hat{\bU} \hat{\bH} - \bU) ( \hat{\bU} \hat{\bH} - \bU )^T \|_F\leq \| \hat{\bU} \hat{\bH} - \bU \|_2 \| \hat{\bU} \hat{\bH} - \bU \|_F \notag \\
	& \leq ( \| \hat{\bU} \bH - \bU \|_2 + \| \hat{\bU} (\hat{\bH} - \bH) \|_2 ) \| \hat{\bU} \hat{\bH} - \bU \|_F \notag \\
	& = ( \| ( \hat{\bP} - \bP ) \bU \|_2 + \| \hat{\bU} (\hat{\bH} - \bH) \|_2 ) \| \hat{\bU} \hat{\bH} - \bU \|_F \notag\\
	&\leq  ( \|  \hat{\bP} - \bP \|_2 + \| \hat{\bH} - \bH \|_2 ) \| \hat{\bU} \hat{\bH} - \bU \|_F
	\notag \\&
	\leq \left( \frac{1}{1-\varepsilon} + \frac{1}{16}  \right) \varepsilon \| \hat{\bU} \hat{\bH} - \bU \|_F
	\leq 1.18  \varepsilon \| \hat{\bU} \hat{\bH} - \bU \|_F
	,\label{2ndDK-bound-6}
	\end{align}
	where the penultimate inequality uses $\|  \hat{\bP} - \bP \|_2 \leq \varepsilon/(1-\varepsilon)$ and (\ref{2ndDK-bound-H}).
	By defining $\bW = \hat{\bU} \hat{\bH} - \bU- f(\bE\bU)$ we can write
	\begin{align}
	(\hat{\bU} \hat{\bH} - \bU) \bU^T + \bU (\hat{\bU} \hat{\bH} - \bU)^T
	=[ f(\bE\bU) \bU^T + \bU f(\bE \bU)^T ] + ( \bW \bU^T + \bU \bW^T )  .
	\label{2ndDK-bound-7}
	\end{align}
	It is easily seen that
	\begin{align*}
	&\| \bW \bU^T + \bU \bW^T \|_F =  \| ( \bP_{\perp} \bW \bU^T  + \bP \bW \bU^T  ) +
	(   \bU \bW^T \bP_{\perp} +  \bU^T \bW \bP )\|_F\notag\\
	&\leq  \|  \bP_{\perp} \bW \bU^T   +    \bU \bW^T \bP_{\perp} \|_F
	+ \|  \bP \bW \bU^T  +  \bU^T \bW \bP \|_F\notag\\
	& = \left(  \| \bP_{\perp} \bW \bU^T \|_F^2 + \|  \bU \bW^T \bP_{\perp} \|_F^2  \right)^{1/2}
	+\|  \bP \bW \bU^T  +  \bU^T \bW \bP \|_F\notag\\
	& \leq \sqrt{2} \| \bP_{\perp} \bW \|_F + 2 \| \bP \bW \|_F.
	\end{align*}
	On the one hand, (\ref{2ndDK-bound-4}) forces that $\| \bP_{\perp} \bW \|_F  \leq \| \bW\|_F \leq \frac{116}{25}\varepsilon \| \hat{\bU} \hat{\bH} - \bU \|_F$. On the other hand, the fact $\bP f(\bE \bU) =\mathbf{0}$, (\ref{2ndDK-bound-1.5}) and (\ref{2ndDK-bound-H}) yield
	\[
	\| \bP\bW \|_F = \| \bP [ \hat{\bU} \hat{\bH} - \bU - f(\bE\bU) ] \|_F = \| \bP \hat{\bU} \hat{\bH} - \bU \|_F \leq
	\sqrt{\frac{5}{16}} \varepsilon \| \hat{\bU} \hat{\bH} - \bU\|_F.
	\]
	Hence
	\begin{align}
	&\| \bW \bU^T + \bU \bW^T \|_F  \leq \left( \frac{116}{25}\sqrt{2} +2\sqrt{\frac{5}{16}} \right) \varepsilon \| \hat{\bU} \hat{\bH} - \bU\|_F \leq 7.68 \varepsilon \| \hat{\bU} \hat{\bH} - \bU\|_F.
	\label{2ndDK-bound-8}
	\end{align}
	By collecting (\ref{2ndDK-bound-5}), (\ref{2ndDK-bound-6}), (\ref{2ndDK-bound-7}) and (\ref{2ndDK-bound-8}) we derive that
	\begin{align*}
	&\| \hat{\bP} -  \bP  - [ f(\bE \bU) \bU^T + \bU f (\bE\bU)^T ] \|_F  \leq 8.86 \varepsilon \| \hat{\bU} \hat{\bH} - \bU\|_F \notag\\
	& \leq 8.86 \varepsilon \frac{\| \hat{\bP} - \bP \|_F }{ \sqrt{2 - \| \hat{\bP} - \bP \|_2^2 } }
	\leq 8.86 \varepsilon \frac{\| \hat{\bP} - \bP \|_F }{ \sqrt{2 -1/9^2 } }
	\leq 6.29 \varepsilon \| \hat{\bP} - \bP \|_F,
	\end{align*}
	where we also used Lemma \ref{2ndDK-lem-projection} and $\| \hat{\bP} - \bP \|_2 \leq \frac{\varepsilon}{1-\varepsilon} \leq 1/9$. Therefore,
	\begin{align*}
	&\frac{ \| f(\bE \bU) \bU^T + \bU f (\bE\bU)^T  \|_F }{1 + 6.29\varepsilon} \leq \| \hat{\bP} -  \bP \|_F \leq  \frac{\| f(\bE \bU) \bU^T + \bU f (\bE\bU)^T  \|_F}{1-6.29\varepsilon},\notag\\
	&\| \hat{\bP} -  \bP  - [ f(\bE \bU) \bU^T + \bU f (\bE\bU)^T ] \|_F  \leq \frac{6.29\varepsilon}{1-6.29\varepsilon}\| f(\bE \bU) \bU^T + \bU f (\bE\bU)^T  \|_F \notag\\
	& \leq 16.96 \varepsilon
	\| f(\bE \bU) \bU^T + \bU f (\bE\bU)^T  \|_F.
	\end{align*}
	We finish the proof by
	\begin{align*}
	&\| f(\bE \bU) \bU^T + \bU f (\bE\bU)^T  \|_F^2\\
	&=\| f(\bE \bU) \bU^T \|_F^2 + \| \bU f (\bE\bU)^T  \|_F^2 + 2 \Tr \left(
	[ f(\bE \bU) \bU^T ]^T \bU f (\bE\bU)^T
	\right)\\
	&= \| f(\bE \bU) \bU^T \|_F^2 + \| \bU f (\bE\bU)^T  \|_F^2 + 0 = 2 \| f(\bE\bU)\|_F^2.
	\end{align*}
\end{proof}

\end{document}